\RequirePackage{fix-cm}
\RequirePackage{amsmath}
\documentclass[smallcondensed]{svjour3} 
\smartqed  
\usepackage{graphicx}
\usepackage{latexsym,enumerate,amssymb,amsbsy}
\usepackage{graphics,color}
\usepackage{verbatim}
\usepackage{url}
\usepackage{dblfloatfix}
\usepackage{mathptmx}
\usepackage{tikz}
\usepackage{xfrac}
\usepackage{amsfonts}
\usepackage[noadjust]{cite}
\smartqed
\usepackage{algorithm}
\usepackage{algpseudocode}
\usepackage{enumitem}
\newcommand{\edit}[1]{{\color{red}({\bf Note}: #1)}}

\newcommand{\Mod}[1]{\ (\textup{mod}\ #1)}
\newcommand{\F}{\mathbb F}
\newcommand{\N}{\mathbb N}
\newcommand{\Z}{\mathbb Z}
\def\v{{\vec{v}}}
\def\S{{\vec{S}}}
\def\u{{\vec{u}}}
\def\s{{\vec{s}}}
\def\0{{\vec{0}}}
\def\1{{\vec{1}}}
\def\a{{\vec{a}}}
\def\b{{\vec{b}}}
\def\c{{\vec{c}}}
\def\p{{\vec{p}}}
\def\q{{\vec{q}}}
\def\vX{{\vec{X}}}
\def\O{{\mathcal{O}}}
\def\M{{\mathcal{M}}}
\def\I{{\mathcal{I}}}
\def\P{{\mathcal{P}}}
\def\A{{\mathcal{A}}}
\DeclareMathOperator{\lcm}{lcm}
\journalname{}
\begin{document}
\title{Construction of de Bruijn Sequences from Product of Two Irreducible Polynomials}
\author{Zuling Chang \and Martianus Frederic Ezerman \and\\San Ling \and Huaxiong Wang}
\authorrunning{Chang \and Ezerman \and Ling \and Wang}
\institute{Z. Chang \at School of Mathematics and Statistics, Zhengzhou University, Zhengzhou 450001, China\\
\email{zuling\textunderscore chang@zzu.edu.cn}
\and
M. F. Ezerman \and S. Ling \and H. Wang \at Division of Mathematical Sciences, School of Physical and Mathematical Sciences,\\
Nanyang Technological University, 21 Nanyang Link, Singapore 639798\\
\email{\{fredezerman,lingsan,HXWang\}@ntu.edu.sg}
}
\date{Received: date / Accepted: date}
\maketitle
\begin{abstract}
We study a class of Linear Feedback Shift Registers (LFSRs) with characteristic polynomial $f(x)=p(x)q(x)$ where
$p(x)$ and $q(x)$ are distinct irreducible polynomials in $\F_2[x]$. Important properties of the LFSRs, such as the cycle structure and the adjacency graph, are derived. A method to determine a state belonging to each cycle and a generic algorithm to find all conjugate pairs shared by any pair of cycles are given. The process explicitly determines the edges and their labels in the adjacency graph. The results are then combined with the cycle joining method to efficiently construct a new class of de Bruijn sequences. An estimate of the number of resulting sequences is given. In some cases, using cyclotomic numbers, we can determine the number exactly.
\end{abstract}

\section{Introduction}\label{sec:intro}
A binary {\it de Bruijn sequence} of order $n$ is a binary sequence with period $N=2^n$
in which each $n$-tuple occurs exactly once in one period of the sequence. There are $2^{2^{n-1}-n}$
such sequences~\cite{Bruijn46}.

De Bruijn sequences have been studied for a long time using diverse mathematical tools and often show up in multiple
disguises~\cite{Ral82}. They have many applications in communication systems, coding theory, and cryptography due to their attractive characteristics, such as having long period and large linear complexity, and being balanced~\cite{Chan82,Golomb81}. Fredricksen's survey~\cite{Fred82} discusses their various properties and constructions.

A well-known construction approach called the {\it cycle joining (CJ) method} (see {\it e.g.},~\cite{Fred82} and~\cite{Golomb81}) joins all cycles produced by a given Feedback Shift Register (FSR) into a single cycle. Since the cycle structure of a Linear FSR (LFSR) has been well-studied, it is natural to construct de Bruijn sequences by applying the cycle joining method to LFSRs. Some LFSRs with simple cycle structure, such as the maximal length LFSRs, pure cycling registers, and pure summing registers, have been used to generate de Bruijn sequences using the said method in~\cite{Fred75,Fred82,Etzion84}.

Hauge and Helleseth established a connection between the cycles generated by LFSRs and irreducible cyclic codes in~\cite{HH96}. The number of de Bruijn sequences obtained from these LFSRs is related to cyclotomic numbers. The cycle structure and the adjacency graph of LFSRs with simple reducible polynomials are then studied and several classes of de Bruijn sequences are constructed from these LFSRs.

Recent studies have considered cases where the characteristic polynomials are products of some simple or primitive polynomials. In~\cite{Li14-1}, a class of de Bruijn sequences was derived from LFSRs with characteristic polynomials $(1+x)^3 p(x)$ with $p(x)$ a primitive polynomial of degree $n>2$. In~\cite{Li14-2} the focus was on characteristic polynomials $(1+x^3) p(x)$. The characteristic polynomials studied in~\cite{Li16} are products of primitive polynomials whose degrees are pairwise coprime. Hence, the sequences forming the cycle structure have coprime periods. Although this set up leads to a structure that can be nicely studied and described, in most cases the number of de Bruijn sequences that the construction yields is small when compared with the construction that we are proposing in this paper. An example in Section~\ref{sec:genpoly} will highlight this fact.

In this paper, we construct new de Bruijn sequences based on LFSRs with characteristic polynomials
$f(x)=p(x)q(x)$, where $p(x)$ and $q(x)$ are distinct irreducible polynomials.
We study the corresponding cycle structure and construct the adjacency graph. We propose
a method to find a set of representatives of the states, one belonging to each cycle, and design an algorithm to find
all conjugate pairs shared by any two cycles. Deploying the cycle joining method, we construct
the de Bruijn sequences and estimate their number. In some instances, the estimates are made exact.

This work contributes to the large literature on de Bruijn sequences on several fronts. We generalize the choices of characteristic polynomials to products of irreducible polynomials, instead of those of primitive polynomials. The structure of the resulting LFSRs is thoroughly studied. Our step-by-step construction of de Bruijn sequences from the LFSRs remains efficient to perform while handling more complex cycle structure,
yielding a large number of de Bruijn sequences. The resulting class contains many known ones as special cases. In particular, the class derived from product of two primitive polynomials is a subclass of our construction.
Finally, most of the methods developed in this paper generalize naturally to LFSRs with product of $s>2$ pairwise distinct irreducible polynomials as characteristic polynomials.

The paper is organized as follows. After this introduction come preliminary notions and known results in Section~\ref{sec:prelims}. Section~\ref{sec:cycle} presents the cycle structure. The main results are presented in Section~\ref{sec:main} in two parts. The first part determines the adjacency graph. The second part provides an algorithm to find all conjugate pairs between any two cycles and gives
a rough estimate of the number of constructed de Bruijn sequences. A detailed example in Section~\ref{sec:ex} showcases how the theoretical results fit together nicely in practice. Section~\ref{sec:spcases} examines three special cases where the characteristic polynomial
has certain simplifying properties. Section~\ref{sec:genpoly} briefly treats a more general case where the characteristic polynomial is the product of more than two irreducible polynomials. The last section contains a brief conclusion and some future directions.

\section{Preliminaries}\label{sec:prelims}
We use~\cite[Chapter~4]{GG05} as a main reference for this section.

An {\it $n$-stage shift register} is a circuit consisting of $n$ consecutive storage units, each containing a bit, regulated by a clock. As the clock pulses, the bit in each storage unit is shifted to the next stage in line. A shift register becomes a binary code generator when one adds a feedback loop which outputs a new bit $s_n$ based on the $n$ bits $\s_0= (s_0,\ldots,s_{n-1})$ called an {\it initial state} of the register. The corresponding {\it feedback function} $f(x_0,\ldots,x_{n-1})$ is the Boolean function that outputs $s_n$ on input $\s_0$.

A feedback shift register (FSR) outputs a binary sequence $\s=s_0,s_1,\ldots,s_n,\ldots$ satisfying the recursive relation $s_{n+\ell} = f(s_{\ell},s_{\ell+1},\ldots,s_{\ell+n-1})$ for $\ell = 0,1,2,\ldots$. For $N \in \N$, if $s_{i+N}=s_i$ for all $i \geq 0$, then $\s$ is {\it $N$-periodic}
or {\it with period $N$} and one writes $\s= (s_0,s_1,s_2,\ldots,s_{N-1})$.
We call $\s_i= (s_i,s_{i+1},\ldots,s_{i+n-1})$ {\it the $i$-th state} of $\s$ and states $\s_{i-1}$ and $\s_{i+1}$ the {\it predecessor} and {\it successor} of $\s_i$, respectively.

Given two sequences $\u= u_0,u_1,\ldots$ and $\v =v_0,v_1,\ldots$, the sum $\u + \v$ and the scalar multiple $c \u$ are
$\u+\v = u_0+v_0,u_1+v_1,\ldots$ and $c\u =c u_0, c u_1,\ldots$. A period of the sum is the least common multiple ($\lcm$) of
the periods of the given sequences.

For an FSR, distinct initial states generate distinct sequences. We collect all these sequences to form a set $\Omega(f)$ of cardinality $2^n$. All sequences in $\Omega(f)$ are periodic if and only if the feedback function $f$ is {\it nonsingular}, i.e., $f$ can be written as
\[
f(x_0,x_1,\ldots,x_{n-1})=x_0+g(x_1,\ldots,x_{n-1}),
\]
where $g(x_1,\ldots,x_{n-1})$ is some Boolean function with domain $\F_2^{n-1}$~\cite[page~116]{Golomb81}. In this paper, the feedback functions are all nonsingular. An FSR is called {\it linear} or an LFSR if its feedback function is linear, and {\it nonlinear} or an NLFSR otherwise.

The {\it characteristic polynomial} of an $n$-stage LFSR with feedback function
\[
f(x_0,x_1,\ldots,x_{n-1})= \sum_{i=0}^{n-1} c_i x_i
\]
is the polynomial $f(x)=x^n+\sum_{i=0}^{n-1}c_ix^i \in \F_2[x]$.
A sequence $\s$ may have many characteristic polynomials.
We call the monic characteristic polynomial with the lowest degree the {\it minimal polynomial} of $\s$.
It represents the LFSR of shortest length that generates $\s$.
More properties of the minimal polynomial can be found in~\cite[Sections 4.2 and 4.3]{GG05}.
For an LFSR with characteristic polynomial $f(x)$, the set $\Omega(f)$ is also denoted by $\Omega(f(x))$.
\begin{example}
A $3$-stage NLFSR with initial state $(110)$ and feedback function $f(x_0,x_1,x_2)=x_0+x_1 x_2 + x_2 + 1$ outputs $(1100~0101)$, a de Bruijn sequence with period $8$.
\end{example}

For a sequence $\s$, the {\it (left) shift operator} $L$ is given by
\[
L\s=L(s_0,s_1,\ldots,s_{N-1})=(s_1,s_2,\ldots,s_{N-1},s_0)
\]
with the convention that $L^0\s=\s$. The set $[\s]:=\{\s,L\s,L^2\s,\ldots,L^{N-1}\s \}$ is a {\it shift equivalent class} or a {\it cycle} in $\Omega(f)$. The set of sequences in $\Omega(f)$ can be partitioned into cycles.

If $\Omega(f(x))$ consists of exactly $r$ cycles $[\s_1],[\s_2],\ldots, [\s_r]$ for some $r \in \N$, then the {\it cycle structure} of $\Omega(f(x))$ is
\[
\Omega(f(x))=[\s_1] \cup [\s_2] \cup \ldots \cup [\s_r].
\]
When $r=1$, the corresponding FSR is of {\it maximal length} and its output sequences are de Bruijn sequences of order $n$. A nonzero output sequence of a maximal length $n$-stage LFSR is said to be an {\it $m$-sequence of order $n$} or a {\it maximal length sequence} (MLS).

A state $\v=(v_0,v_{1},\ldots,v_{n-1})$ and its {\it conjugate} $\widehat{\v}=(v_0 + 1,v_{1},\ldots,v_{n-1})$ form a {\it conjugate pair}. Cycles $C_1$ and $C_2$ are {\it adjacent} if they are disjoint and there exists $\v$ in $C_1$ whose conjugate $\widehat{\v}$ is in $C_2$.

Adjacent cycles $C_1$ and $C_2$ with the same feedback function $g(x_0,x_1,\ldots,x_{n-1})$ can be joined into a single cycle by interchanging the successors of $\v$ and $\widehat{\v}$. The corresponding feedback function of the resulting cycle is
\[
h(x_0,x_1,\ldots,x_{n-1})=g(x_0,x_1,\ldots,x_{n-1})+\prod_{i=1}^{n-1}(x_i+v_i+1).
\]

The basic idea in the cycle joining method is to provide the feedback functions of the new de Bruijn sequences by finding the corresponding
conjugate pairs. Determining the conjugate pairs between cycles is, therefore, a crucial step in constructing de Bruijn sequences.
\begin{definition}~\cite{Hauge96}
For an FSR with feedback function $f$, its {\it adjacency graph} $G$ is an undirected multigraph whose vertices correspond to the cycles in $\Omega(f)$. There exists an edge between two vertices if and only if they share a conjugate pair. The number of shared conjugate pairs labels the edge.
\end{definition}

When the edges connecting two vertices are considered pairwise distinct, there is a one-to-one correspondence between the spanning trees of the adjacency graph $G$ and the de Bruijn sequences constructed by the CJ method. The details can be found in~\cite{HH96} and~\cite{Hauge96}. The following result, a variant of the BEST (de {\bf B}ruijn, {\bf E}hrenfest, {\bf S}mith, and {\bf T}utte) Theorem adapted from~\cite[Section~7]{AEB87}, provides the counting formula.
\begin{theorem}(BEST)\label{BEST} Let $G$ be the adjacency graph of an FSR with vertex set $\{v_1,v_2,\ldots,v_{\ell}\}$. Let $G'$ be the graph obtained by removing all loops in $G$. Let $\M=(m_{i,j})$ be the $\ell \times \ell$ matrix derived from $G'$ in which $m_{i,i}$ is the sum of the labels on the edges incident to $v_i$ and $m_{i,j}$ is the negative of the label of edge $(v_i,v_j)$ for $i \neq j$. Then the number of the spanning trees of $G$ is the cofactor of any entry of $\M$.
\end{theorem}
The {\it cofactor} of the entry $m_{i,j}$ in $\M$ is $(-1)^{i+j}$ times the determinant of the matrix obtained by deleting the $i$-th row and $j$-th column of $\M$. Relevant concepts and results on finite fields, such as the definitions and properties of minimal, irreducible, and primitive polynomials, can be found in \cite{LN97}.

With the preparatory notions in place, we proceed to determining the cycle structure.

\section{The Cycle Structure of $\Omega(f(x))$}\label{sec:cycle}
We start by recalling some useful properties and results.

Let $g(x)\in \F_2[x]$ be an irreducible polynomial of degree $n$ having $\beta \in\F_{2^n}$ as a root. Then there exists
a primitive element $\alpha\in\F_{2^n}$ such that $\beta=\alpha^t$ for some $t \in \N$, and $e=\frac{2^n-1}{t}$ is
the order of $\beta$. Using the {\it Zech logarithmic representation} (see. {\it e.g.}, \cite[page~39]{GG05}), we write
\[
1+ \alpha^{\ell} = \alpha^{\tau_{n}(\ell)}
\]
where $\tau_n (\ell)$ is the Zech logarithm relative to $\alpha$ that induces
a permutation on $\{1,2,\ldots,2^n-2\}$. For completeness, $\tau_n (\ell):=\infty$
for $\ell \equiv 0 \Mod{2^n-1}$ and $\alpha^{\infty}:=0$.

The {\it cyclotomic classes} $C_i\subseteq\F_{2^n}$ for $0 \leq i <t$ are
\begin{equation}\label{eq:cyclas}
C_i=\{\alpha^{i+ s \cdot t}~|~0\leq s <e\}=\{\alpha^i\beta^s~|~0\leq s<e\}=\alpha^i C_0.
\end{equation}
The {\it cyclotomic numbers} $(i,j)_{t}$, for $0\leq i,j <t$ are given by
\begin{equation}\label{eq:cycnum}
(i,j)_{t} =\left|\{(\xi,\xi+1)~|~\xi\in C_i, \xi+1\in C_j\}\right|=\left|\{\xi~|~\xi\in C_i, \xi+1\in C_j\}\right|.
\end{equation}
Requiring $\xi\in C_i$ and $\xi+1\in C_j$ is equivalent to requiring that there exist $s$ and $s'$ with $0 \leq s,s' <e$ such that
\[
1+\alpha^{i + s \cdot t} = \alpha^{j + s' \cdot t} \iff \tau_n(i + s \cdot t) = j + s' \cdot t
\iff \tau_n(i + s \cdot t) \equiv j \Mod{t}.
\]
Thus, an equivalent expression to~(\ref{eq:cycnum}) is
\begin{equation}\label{eq:cycnum2}
(i,j)_{t} = \left| \{s~|~\tau_n(i + s \cdot t) \equiv j \Mod{t}\}\right|.
\end{equation}
\begin{remark}\label{rem:cn}
In general, it is hard to determine the cyclotomic numbers for all parameter sets. They are known for small parameters or under certain conditions. Some useful facts can be found in~\cite{Storer67} and~\cite[Section~1.4]{Ding14}. The cyclotomic numbers used in this paper are all known.
\end{remark}

Using $\{1,\beta,\ldots,\beta^{n-1}\}$ as a basis for $\F_{2^n}$ as an $\F_2$-vector space,
for $0\leq j < 2^n-1$, one can uniquely express $\alpha^j$ as
\[
\alpha^j=\sum_{i=0}^{n-1}a_{j,i}\beta^i \text{ with } a_{j,i}\in \F_2.
\]

Define the mapping $\varphi:\F_{2^n}\rightarrow \F_2^{n}$ by
\[
\varphi(0)=\0,\ \ \varphi(\alpha^{j})=(a_{j,0},a_{j+t,0},\ldots,a_{j+(n-1)t,0}),
\]
where the subscripts are reduced modulo $2^n-1$.
Let
\begin{equation}\label{eq:corres}
\u_i = (a_{i,0},a_{i+t,0},\ldots,a_{i+(e-1)t,0}).
\end{equation}
It is shown in~\cite[Theorem 3]{HH96} that the class $C_i$ corresponds to the cycle $[\u_i]$ under the mapping $\varphi$. In other words, $\u_i$ and the sequence of states $((\u_i)_0,(\u_i)_1,\ldots,(\u_i)_{e-1})$ of $\u_i$ where, for $0\leq j<e$,
\[
(\u_i)_j=(a_{i+jt,0},a_{i+(j+1)t,0},\ldots,a_{i+(j+n-1)t,0})=\varphi(\alpha^i\beta^j),
\]
are equivalent. Hence, $\u_i \longleftrightarrow C_i$.

The theory of~LFSRs in~\cite[Chapter~4]{GG05} tells us that
\begin{equation}\label{equ:g}
\Omega(g(x))=[\0]\cup[\u_0]\cup[\u_1]\cup\ldots\cup[\u_{t-1}].
\end{equation}

If $g(x)$ is a primitive polynomial, then $e=2^n-1$ and there exists only one cyclotomic class.
Hence, $\Omega(g(x))=[\0]\cup[\u]$, where
$\u$ is the $m$-sequence with period $2^n-1$.
The sequence $\u$ has the following {\it shift-and-add property}.
\begin{lemma}\cite[Theorem 5.3]{GG05}\label{lem:saa}
Let $\u$ be an $m$-sequence with period $2^n-1$. Then, for $0 < i < 2^n-1$, there exists
$0 < j < 2^n-1$ such that $\u + L^{i} \u = L^{j} \u $ with $j=\tau_n(i)$.
\end{lemma}

When $g(x)$ is not primitive, the situation is more involved.
\begin{lemma}\label{lem:irr-saa}
Let $g(x) \in \F_2[x]$ be an irreducible polynomial of degree $n$ and order $e$ (making $ t=\sfrac{(2^n-1)}{e}$) with $\Omega(g(x))$ as presented in (\ref{equ:g}).
Then, for each triple $(i,j,k)$ with $0 \leq i, j, k < t$, we have
\begin{equation}\label{equ:irr-saa}
(j-i,k-i)_t=\left|\{a~|\u_i+L^a\u_j=L^b\u_k;0\leq a,b <e\}\right|.
\end{equation}
\end{lemma}
\begin{proof}
Using the correspondence
\[
\u_i \longleftrightarrow (\varphi(\alpha^i\beta^0),\varphi(\alpha^i\beta^1),\ldots,\varphi(\alpha^i\beta^{e-1}))=\varphi(\alpha^iC_0),
\]
we have
\[
\u_i + L^a \u_j \longleftrightarrow \varphi(\alpha^iC_0)+\varphi(\alpha^j\beta^aC_0)
=\varphi((\alpha^i+\alpha^j\beta^a)C_0)=\varphi((1+\alpha^{j-i}\beta^a)\alpha^iC_0).
\]
Observe that as $a$ runs through $\{0,1,\ldots,e-1\}$ there are $(j-i,k-i)_t$ such $a$, each of which satisfies $(1+\alpha^{j-i}\beta^a)=\alpha^{k-i}\beta^b$ for some $b$. In each occasion,
\[
\varphi((1+\alpha^{j-i}\beta^a)\alpha^iC_0)=\varphi(\alpha^{k-i}\beta^b\alpha^iC_0)=\varphi(\alpha^k\beta^bC_0).
\]
Note that the corresponding sequences are shifts of $\u_k$ and the proof is now complete.\qed
\end{proof}

\begin{lemma}\cite[Lemma 4.2]{GG05}\label{lemma1}
Let $g(x),h(x)\in \F_2[x]$ be two nonzero polynomials. Denote by $\Omega(g(x))+\Omega(h(x))$
the set of sequences $\{{\bf g}+{\bf h}~|~{\bf g}\in\Omega(g(x)),{\bf h}\in\Omega(h(x))\}$. Then
\begin{enumerate}
\item $\Omega(g(x)) \subseteq \Omega(h(x))$ if and only if $g(x) \mid h(x)$.
\item $\Omega(g(x))+\Omega(h(x))=\Omega(\lcm(g(x),h(x)))$.
\item $\Omega(g(x)) \cap \Omega(h(x)) = \Omega(\gcd(g(x),h(x)))$.
\end{enumerate}
\end{lemma}

\begin{lemma}\label{lem:cycle-f}
Let $f(x)=p(x)q(x)$ where $p(x)$ and $q(x)$ are two distinct irreducible polynomials in $\F_2[x]$ of degree $m$ and $n$ and order $e_1$ and $e_2$, respectively. Let $t_1=\frac{2^m-1}{e_1}$ and
$t_2=\frac{2^n-1}{e_2}$. The cycle structure of $\Omega(f(x))$ is
\begin{equation}\label{equ:cycle-f}
[\0] ~\cup ~ \bigcup_{i=0}^{t_1-1}[\u_i] ~\cup ~\bigcup_{j=0}^{t_2-1}[\s_j] ~\cup~ \left(\bigcup_{i=0}^{t_1-1}\bigcup_{j=0}^{t_2-1}\bigcup_{k=0}^{~\gcd(e_1,e_2)-1~}[L^k\u_i+\s_j]\right).
\end{equation}
\end{lemma}
\begin{proof}
Based on~(\ref{equ:g}), we have
\[
\Omega(p(x))=[\0]\cup[\u_0]\cup[\u_1]\cup\ldots\cup[\u_{t_1-1}]\text{ and }
\Omega(q(x))=[\0]\cup[\s_0]\cup[\s_1]\cup\ldots\cup[\s_{t_2-1}].
\]

By Lemma~\ref{lemma1}, $\Omega(f(x))$ contains $\Omega(p(x))$ and $\Omega(q(x))$ as subsets. Hence,
\[
[\0]~\cup~\bigcup_{i=0}^{t_1-1}[\u_i]~\cup~\bigcup_{j=0}^{t_2-1}[\s_j]\subseteq\Omega(f(x)).
\]
The minimal polynomial of all other sequences in $\Omega(f(x))$ must be $f(x)$. The
period of these sequences is the order of $f(x)$, which is $\lcm(e_1,e_2)$. The sequences are of the form
\[
L^k \u_i + L^{\ell} \s_j = L^{\ell} (L^{k- \ell} \u_i + \s_j)
\]
for some $i,j,k$, and $\ell$, where $k - \ell$ is computed modulo $e_1$. They can be partitioned into
\begin{equation}\label{eq:eqcl}
\frac{2^{m+n}-(2^m+2^n-1)}{\lcm(e_1,e_2)}= \frac{(2^m-1)(2^n-1)}{\lcm(e_1,e_2)}= \frac{(e_1 \cdot t_1)(e_2 \cdot t_2)}{\lcm(e_1,e_2)}=t_1 \cdot t_2 \cdot \gcd(e_1,e_2)
\end{equation}
shift inequivalent classes.

Next, we show that $L^k\u_i+\s_j$ and $L^{k+t \cdot \gcd(e_1,e_2)}\u_i+\s_j$ are shift equivalent for $0 \leq k < \gcd(e_1,e_2)$ and $0<t<\sfrac{e_1}{\gcd(e_1,e_2)}$. Since
$\sfrac{e_1}{\gcd(e_1,e_2)}$ and $\sfrac{e_2}{\gcd(e_1,e_2)}$ are coprime, there exist
$v,w \in \Z$ such that
\[
v~\frac{e_1}{\gcd(e_1,e_2)} + w~\frac{e_2}{\gcd(e_1,e_2)}=1 \iff
t \left( \gcd(e_1,e_2)- v \cdot e_1 \right)=  t \cdot w \cdot e_2.
\]
Since the periods of $\u_i$ and $\s_j$ are, respectively, $e_1$ and $e_2$,
\[
L^{k + t \cdot \gcd(e_1,e_2)} \u_i + \s_j = L^{k+t \cdot \gcd(e_1,e_2)-t \cdot v \cdot e_1}\u_i + L^{t \cdot w \cdot e_2}\s_j
=L^{t \cdot w \cdot e_2}(L^k \u_i + \s_j).
\]

Because $0 \leq i < t_1$, $0 \leq j < t_2$, and $0 \leq k < \gcd(e_1,e_2)$, there are at most $t_1 \cdot t_2 \cdot \gcd(e_1,e_2)$ shift inequivalent cycles $[L^k\u_i+\s_j]$ for which the period of each sequence is $\lcm(e_1,e_2)$. Combined with (\ref{eq:eqcl}), we conclude that there are exactly $t_1 \cdot t_2 \cdot \gcd(e_1,e_2)$ shift inequivalent classes and
$[L^k\u_i+\s_j]$ with $0\leq i<t_1$, $0\leq j<t_2$, and $0\leq k<\gcd(e_1,e_2)$ are the cycles.
\qed
\end{proof}

\section{The Main Results}\label{sec:main}
This most technical section contains two subsections. The first one studies the adjacency graph of $\Omega(f(x))$. The second one begins with a method to find a state belonging to a particular cycle and incorporates the results to design an algorithm that finds all conjugate pairs shared by any two cycles. It ends with a heuristic estimate of the number of de Bruijn sequences constructed.

\subsection{The Adjacency Graph of $\Omega(f(x))$}\label{subsec:graph}
Suppose that the special state $\S:=(1,0,\ldots,0) \in \F_2^{m+n}$ is in a given
cycle $[\a]$. Cycles $[\b]$ and $[\c]$ share a conjugate pair if and only if,
for some $i,j,k \in \Z$,
\[
L^i \b + L^j \c = L^k \a \iff \b + L^{j-i} \c = L^{k-i} \a.
\]
This is the basic rule of finding the conjugate pairs shared by the cycles in
$\Omega(f(x))$.

Since the degrees of the minimal polynomials of $\u_i$ and $\s_j$ are all $< m+n$,
neither $\u_i$ nor $\s_j$ can contain $m+n-1$ consecutive $0$s. Thus, for all $i$ and $j$,
$\S \notin [\u_i]$ and $\S \notin [\s_j]$. Hence, $\S \in [L^c \u_a + \s_b]$ for some nonnegative
integers $a,b$, and $c$. From hereon, we use $a,b,$ and $c$ specifically to refer to
the cycle $[L^c \u_a + \s_b]$ that contains the special state $\S$.

\begin{proposition}\label{prop1}
There exist some $a,b,c \in \Z$ such that $[\0]$ and $[L^c\u_a+\s_b]$ are adjacent.
For arbitrary $i$ and $j$, there is no conjugate pair between $[\u_i]$ and $[\u_j]$ and between $[\s_i]$ and $[\s_j]$.
\end{proposition}

The next result determines the number of conjugate pairs between $[\u_i]$ and $[\s_j]$.
\begin{proposition}\label{prop2}
Let $0 \leq i <t_1$ and $0 \leq j <t_2$. Then $[\u_i]$ and $[\s_j]$ share a conjugate pair
if and only if $i=a$ and $j=b$. When this is the case, the conjugate pair is unique.
\end{proposition}
\begin{proof}
$L^k \u_i + \s_j$ is a shift of $L^c \u_a + \s_b$
if and only if $i=a$, $j=b$, and $k \equiv c \Mod{\gcd(e_1,e_2)}$.
By the proof of Lemma~\ref{lem:cycle-f}, for any $\ell \in \Z$, $[L^c \u_a + \s_b]$
and $[L^{c+ \ell \cdot \gcd(e_1,e_2)} \u_a + \s_b]$ are shift equivalent.
Thus, $[\u_a]$ and $[\s_b]$ share a unique conjugate pair.\qed
\end{proof}

We now consider the number of conjugate pairs between $[\u_i]$ and $[L^{\ell} \u_j+\s_k]$.

\begin{proposition}\label{prop3}
Let $0 \leq i,j <t_1$ and $0 \leq \ell < \gcd(e_1,e_2)$. For a given $(i,j)$, the following properties hold.
\begin{enumerate}
\item $[\u_i]$ and $[L^{\ell} \u_j+\s_k]$ share no conjugate pair when $k \neq b$.
\item The sum of the numbers of conjugate pairs between $[\u_i]$ and $[L^{\ell}\u_j+\s_b]$ from $\ell=0$ to $\ell=\gcd(e_1,e_2)-1$ is the cyclotomic number $\delta_1:=(i-j,a-j)_{t_1}$.
\item Suppose that after determining $\u_j + L^{w} \u_i$, for $0 \leq w < e_1$, we have found the $\delta_1$ distinct shifts of $\u_a$, say
\[
L^{k_0}\u_a,L^{k_1}\u_a,\ldots,L^{k_{\delta_{1}-1}}\u_a.
\]
The exact number of conjugate pairs between $[\u_i]$ and $[L^{\ell}\u_j+\s_b]$ is
\begin{equation}\label{equ:Llu}
\left|\{k_v~|~c-k_v\equiv \ell \Mod {\gcd(e_1,e_2)} \mbox{ with } v=0,1,\ldots,\delta_{1}-1\}\right|.
\end{equation}
\end{enumerate}
\end{proposition}

\begin{proof}
The first statement is clear.

Let $k=b$. Consider, for $ 0 \leq w < e_1$ and $0 \leq \ell < \gcd(e_1,e_2)$, the equation
\begin{equation*}
L^{\ell} \u_j + \s_b + L^w \u_i = L^{\ell}(\u_j + L^{w-{\ell}} \u_i) + \s_b.
\end{equation*}
By Lemma~\ref{lem:irr-saa}, for a given $\ell$, as $w$ runs through $\{0,1,\ldots,e_1-1\}$,
there are $\delta_1$ many $(w-\ell)$'s such that $\u_j + L^{w-\ell}\u_i$ is a shift of $\u_a$.
By choosing an appropriate $\ell$, we ensure that $L^{\ell}(\u_j+L^{w-\ell}\u_i)+\s_b$ is a shift of $L^c \u_a+\s_b$.
Thus, for a given pair $(i,j)$ with $0 \leq i,j <t_1$, the sum of the numbers of conjugate pairs between $[\u_i]$ and
$[L^{\ell} \u_j+\s_b]$ from $\ell=0$ to $\ell=\gcd(e_1,e_2)-1$ is the cyclotomic number $\delta_1$.
This proves Statement 2.

Each of $L^{k_0}\u_a, L^{k_1}\u_a, \ldots, L^{k_{\delta_{1}-1}}\u_a$ corresponds to an
$\ell_v \equiv c-k_v \Mod{\gcd(e_1,e_2)}$ with $0 \leq v <\delta_1$. Hence, $L^{\ell_v} L^{k_v} \u_a = L^{c'}\u_a$
where $c'\equiv c \Mod {\gcd(e_1,e_2)}$. Thus, the number of conjugate pairs between $[\u_i]$ and $[L^{\ell} \u_j+\s_b]$
for a given $\ell$ is indeed as given in (\ref{equ:Llu}).
\qed
\end{proof}

In an analogous way, we can obtain similar results on the number of conjugate pairs between
$[\s_i]$ and $[L^{\ell}\u_k+\s_j]$ for $0\leq i,j <t_2$.

\begin{proposition}\label{prop4}
Let $0\leq i,j<t_2$ and $0 \leq \ell < \gcd(e_1,e_2)$. For a given $(i,j)$, the following properties hold.
\begin{enumerate}
\item $[\s_i]$ and $[L^{\ell}\u_k+\s_j]$ share no conjugate pair when $k \neq a$.
\item The sum of the numbers of conjugate pairs between $[\s_i]$ and $[L^{\ell}\u_a+\s_j]$ from $\ell=0$ to $\ell=\gcd(e_1,e_2)-1$ is the cyclotomic number $\delta_2:=(i-j,b-j)_{t_2}$.
\item Suppose that after computing $(\s_j + L^w \s_i)$, for $0 \leq w < e_2$, we have determined the $\delta_2$
distinct shifts of $\s_b$, say
 \[
L^{k_0}\s_b, L^{k_1}\s_b, \ldots, L^{k_{\delta_{2}-1}}\s_b.
\]
The exact number of conjugate pairs between $[\s_i]$ and $[L^{\ell}\u_a+\s_j]$ is
\begin{equation}\label{equ:Lls}
\left|\{k_v~|~c+k_v\equiv \ell \Mod{\gcd(e_1,e_2)} \mbox{ with } v=0,1,\ldots,\delta_{2}-1\}\right|.
\end{equation}
\end{enumerate}
\end{proposition}

Let  $0 \leq i_1, i_2 <t_1$, $0 \leq j_1, j_2 < t_2$, and $0 \leq \ell_1, \ell_2 < \gcd(e_1,e_2)$.
We determine the number of conjugate pairs between $[L^{\ell_1}\u_{i_1}+\s_{j_1}]$ and $[L^{\ell_2}\u_{i_2}+\s_{j_2}]$, using
$\lambda:=(j_2-j_1,b-j_1)_{t_2}$ and $\mu:=(i_2-i_1,a-i_1)_{t_1}$ for brevity. Based on
Lemma~\ref{lem:irr-saa}, we know of the following facts.

\begin{enumerate}[label=Fact \arabic*:,leftmargin=2.1\parindent]
\item $L^{k_0}\s_b, L^{k_1}\s_b, \ldots, L^{k_{\lambda-1}}\s_b$ are the $\lambda$
distinct shifts of $\s_b$ generated from $\s_{j_1}+L^{\ell}\s_{j_2}$.
We denote the corresponding $\ell$'s modulo $e_2$ by $c_0, c_1,\ldots, c_{\lambda-1}$.
\item $L^{k'_0}\u_a,L^{k'_1}\u_a,\ldots,L^{k'_{\mu-1}}\u_a$ are the $\mu$
distinct shifts of $\u_a$ generated from $\u_{i_1}+L^{\ell}\u_{i_2}$. We denote
the corresponding $\ell$'s modulo $e_1$ by $d_0,d_1,\ldots,d_{\mu-1}$.
\end{enumerate}

\begin{proposition}\label{prop5}
With $\lambda$ and $\mu$ as given above, let $0 \leq i_1,i_2 <t_1$, $0 \leq j_1,j_2 <t_2$, and $0 \leq \ell_1,\ell_2 < \gcd(e_1,e_2)$. For a given $(i_1,i_2,j_1,j_2)$-tuple, the following properties hold.
\begin{enumerate}
\item The sum of the numbers of conjugate pairs between cycles $[L^{\ell_1}\u_{i_1}+\s_{j_1}]$ and $[L^{\ell_2}\u_{i_2}+\s_{j_2}]$ over all possible $\ell_1$ and $\ell_2$ is $\lambda \cdot \mu$.
\item The exact number of conjugate pairs between two \underline{distinct} cycles $[L^{\ell_1}\u_{i_1}+\s_{j_1}]$ and $[L^{\ell_2}\u_{i_2}+\s_{j_2}]$ is
\begin{align}\label{equ:twocyc}
|\{
(d_i,k'_i,c_j,k_j)~|~
\ell_1 &\equiv c+k_j-k'_i \Mod{\gcd(e_1,e_2)}\text{ and } \notag \\
\ell_2 &\equiv c+k_j-k'_i + d_i - c_j \Mod{\gcd(e_1,e_2)} \text{, with }\notag\\
0 &\leq i < \mu \text{ and } 0 \leq j < \lambda\}|.
\end{align}
\end{enumerate}

Let $0 \leq i <t_1$, $0 \leq j <t_2$, and $0 \leq \ell < \gcd(e_1,e_2)$. The exact number of conjugate pairs between $[L^{\ell}\u_{i}+\s_{j}]$ and itself is
\begin{align}\label{equ:onecyc}
\frac{1}{2}
|\{
(d_{i},k'_{i},c_{j},k_{j})~|~
d_{i} &\equiv c_{j} \Mod{\gcd(e_1,e_2)} \text{ and } \ell\equiv c+k_{j}-k'_{i}\Mod{\gcd(e_1,e_2)} \notag\\
\text{with } 0 & \leq i <(0,a-i)_{t_1} \text{ and } 0 \leq j < (0,b-j)_{t_2}\}|.
\end{align}
\end{proposition}

\begin{proof}
Let $C_1=[L^{\ell_1}\u_{i_1}+\s_{j_1}]$ and $C_2=[L^{\ell_2}\u_{i_2}+\s_{j_2}]$.
Let $0 \leq \ell < \lcm(e_1,e_2)$ and consider
\begin{align}\label{equ:usus1}
L^{\ell_1}\u_{i_1}+\s_{j_1}+L^{\ell}(L^{\ell_2}\u_{i_2}+\s_{j_2})&=(L^{\ell_1}\u_{i_1}+L^{\ell+ \ell_2}\u_{i_2})+(\s_{j_1}+L^{\ell}\s_{j_2})\notag\\
&=L^{\ell_1}(\u_{i_1}+L^{\ell+ \ell_2- \ell_1}\u_{i_2})+(\s_{j_1}+L^{\ell}\s_{j_2}).
\end{align}
To guarantee that this sequence is a shift of $L^c\u_a+\s_b$, we must ensure that
$\u_{i_1}+L^{\ell+ \ell_2- \ell_1}\u_{i_2}$ is a shift of $\u_a$ and
$\s_{j_1}+L^{\ell}\s_{j_2}$ is a shift of $\u_b$.
By Fact 1, $\ell$ must satisfy the system of congruences
\begin{equation}\label{equ:usus2}
\begin{cases}
\ell \equiv d_i + \ell_1- \ell_2 \Mod{e_1}~|~0\leq i< \mu,\\
\ell \equiv c_j \Mod{e_2}~|~0 \leq j < \lambda.
\end{cases}
\end{equation}

By the Chinese Remainder Theorem \cite[Theorem 2.9]{Nath00}, the system has
a unique solution if and only if, modulo $\gcd(e_1,e_2)$,
\begin{equation}\label{equ:usus3}
c_j \equiv d_i + \ell_1-\ell_2 \iff
\ell_2 - \ell_1 \equiv d_i-c_j.
\end{equation}

If $\ell_1$ and $\ell_2$ satisfy~(\ref{equ:usus3}) and $\ell$ satisfies~(\ref{equ:usus2}),
then~(\ref{equ:usus1}) can be expressed as
\begin{equation}\label{equ:usus4}
L^{\ell_1} L^{k'_i} \u_a + L^{k_j} \s_{b}=L^{k_j}(L^{\ell_1+k'_i-k_j}\u_a+\s_b).
\end{equation}
Computing modulo $\gcd(e_1,e_2)$ and taking
\begin{equation}\label{equ:usus5}
\ell_1 \equiv c+k_j-k'_i \text{ and } \ell_2 \equiv c+k_j-k'_i + d_i - c_j,
\end{equation}
the sequence in~(\ref{equ:usus4}) is indeed the required shift of $L^c \u_a+ \s_b$. Thus, we get a conjugate pair between $[L^{\ell_1}\u_{i_1}+\s_{j_1}]$ and $[L^{\ell_2}\u_{i_2}+\s_{j_2}]$. Note that, as $\ell_1$ and $\ell_2$ range through all of their respective values, it may happen that $C_1=C_2$ for some $(\ell_1,\ell_2)$ combination. When this is the case, we count the conjugate pairs $(\v,\hat{\v})$ and $(\hat{\v},\v)$ separately even though they are the same. There are $\mu \cdot \lambda$ choices for the tuple $(d_i,k'_i,c_j,k_j)$, proving Statement 1.

To verify Statement 2, notice that the exact number of conjugate pairs between two distinct cycles $[L^{\ell_1}\u_{i_1}+\s_{j_1}]$ and $[L^{\ell_2}\u_{i_2}+\s_{j_2}]$ is equal to the number of
the tuples $(d_i,k'_i,c_j,k_j)$ that satisfy (\ref{equ:usus5}).

It remains to count the exact number of conjugate pairs between $[L^{\ell}\u_{i}+\s_{j}]$ and itself.
By (\ref{equ:usus3}) and (\ref{equ:usus5}), computing modulo $\gcd(e_1,e_2)$, we have $d_{i}\equiv c_{j}$ and $\ell\equiv c+k_{j}-k'_{i}$. When considering the conjugate pairs between a cycle and itself, every conjugate pair is double counted. To get the correct number we halve the count.
\qed
\end{proof}

\begin{theorem}\label{main}
The adjacency graph of $\Omega(f(x))$ can be constructed based on the results in Propositions~\ref{prop1} to~\ref{prop5}.
\end{theorem}

In the earlier process of computing the number of conjugate pairs, we emphasize
the \emph{ordering of the cycles} by specifying the parameters $i,j$ and $\ell$ in $[L^{\ell} \u_i+\s_j]$. The main reason is to benefit from the notions of cyclotomic classes and numbers. We also require $\S=(1,0,\ldots,0) \in [L^c \u_a + s_b]$. In practice, however, the order of the cycles does not matter. For two distinct orderings of the cycles, the corresponding matrices constructed based on Theorem~\ref{BEST} can be obtained from each other by properly permuting the rows and columns, which does not affect the cofactor.

\subsection{Finding Conjugate Pairs}\label{subsec:Algo}
Recall the definition of $\s_i$ and its successor $\s_{i+1}$ of an $n$-stage FSR sequence $\s$ with feedback function $f(x_0,\ldots,x_{n-1})$ from Section~\ref{sec:prelims}. A {\it state operator} $T$ turns $\s_i$ into $\s_{i+1}$ with $s_{i+n} =
f(s_i,\ldots,s_{i+n-1})$. Hence, if the state $\s_i$ belongs to cycle $[\s]$, then all the states of $[\s]$ are
$\s_i,T \s_i, T^2 \s_i,\ldots $. If $e$ is the period of $\s$, then the distinct states are
\[
\s_i, T \s_i = \s_{i+1},\ldots, T^{e-1} \s_i = \s_{i+e-1}.
\]
Thus, finding one state in a given cycle is sufficient to generate all distinct states. To reduce clutters,
$T$ may be used to denote the state operator for distinct cycles with distinct stages and $\0$ is used to
denote zero vectors and sequences with arbitrary lengths. The context provides enough information to avoid confusion.

For any irreducible polynomial of degree $n$ and order $e$ over $\F_2$, the corresponding cycle structure is given in (\ref{equ:g}). For each cycle, there are
several ways to find one of its states. One can perform an exhaustive search or use the correspondence between cycles and cyclotomic classes
defined in the Section \ref{sec:cycle} to accomplish the task.

Now, assume that a state belonging to each of the cycles in $\Omega(p(x))$ and in $\Omega(q(x))$ has been found.
We propose an efficient way to determine a state belonging to each of the cycles
in $\Omega(f(x)=p(x)q(x))$. Let $[L^{\ell} \u_i+ \s_j]$ with $i,j,\ell \in \Z$ be a cycle in $\Omega(f(x))$.
If $\u_i = u_0,u_1,u_2,\ldots$ and $\s_j = s_0,s_1,s_2,\ldots$, then we have
\[
L^{\ell} \u_i + \s_j = u_{\ell}+s_0, u_{\ell+1} +s_1, u_{\ell+2} + s_2,\ldots
\]
and the $k$-th state $\v_k=(v_0,v_1,\ldots,v_{m+n-1})$ of $[L^{\ell} \u_i+ \s_j]$ satisfies
\begin{align}\label{eq:s-s}
\v_k &=(u_{\ell+k}+s_{k},\ldots,u_{\ell+k+m+n-1}+s_{k+m+n-1})\notag\\
 &=(u_{\ell+k},\ldots,u_{\ell+k+m+n-1})+(s_k,\ldots,s_{k+m+n-1}).
\end{align}
Note that $(u_{\ell+k}, \ldots, u_{\ell+k+m+n-1})$ and $(s_k,\ldots,s_{k+m+n-1})$ are uniquely and
linearly determined by, respectively, $(u_{\ell+k},\ldots,u_{\ell+k+m-1})$ and $(s_k,\ldots,s_{k+n-1})$. These last two are,
respectively, the $m$-stage $(\ell+k)$-th state of $\u_i$ and the $n$-stage $k$-th state of $\s_j$.
Thus, once the states of $\u_i$ and $\s_j$ are known, we can determine the corresponding state in $L^{\ell} \u_i+ \s_j$.

Now, let $\v_k$ be known. Since $(u_{\ell+k},\ldots,u_{\ell+k+m+n-1})$
and $(s_k,\ldots,s_{k+m+n-1})$ are uniquely and linearly determined by $(u_{\ell+k},\ldots,u_{\ell+k+m-1})$
and $(s_k,\ldots,s_{k+n-1})$, respectively, one can use (\ref{eq:s-s}) to construct nonhomogeneous linear equations whose
unique solution is, by the properties of LFSR, $(u_{\ell+k},\ldots,u_{\ell+k+m-1},s_k,\ldots,s_{k+n-1})$. Thus, from $\v_k$,
the $m$-stage $(\ell+k)$-th state of $\u_i$ and the $n$-stage $k$-th state of $\s_j$ can be uniquely determined.

We construct an $(m+n) \times (m+n)$ matrix $P$ from two matrices, namely an $m \times (m+n)$ matrix $P_1$ built from $p(x)$ and an $n \times (m+n)$ matrix $P_2$ based on $q(x)$. $P_1$ is the first $m$
rows of $P$ while $P_2$ is the last $n$ rows.
The $i$-th row of $P_1$ is the first $m+n$ bits of the sequence generated by the LFSR with
characteristic polynomial $p(x)$ whose $m$-stage initial state has $1$ in the $i$-th position and $0$ elsewhere. Similarly, the $j$-th row of $P_2$ is the first $m+n$ bits of the sequence generated by the LFSR with characteristic polynomial $q(x)$ whose $n$-stage initial state has $1$ in the $j$-th position and $0$ elsewhere. Hence, the first $m$ columns of $P_1$ is the $I_m$ identity matrix and the first $n$ columns of $P_2$ is the $I_n$ identity matrix. Since $p(x)$ and $q(x)$ are distinct irreducible polynomials, $P$ is full-rank.

Let $\v \in \F_2^{m+n}$, $\a \in \F_2^m$, and $\b \in \F_2^n$ be, respectively, the initial $(m+n)$-, $m$-, and $n$-stage states of $L^{\ell} \u_i+ \s_j$, $\u_i$, and $\s_j$. We denote by $(\a,\b) \in \F_2^{m+n}$ the simple concatenation of $\a$ and $\b$. There is a one-to-one correspondence between $\v$ and $(T^{\ell} \a,\b)$ through the mapping $P$
\begin{equation}\label{trans:V-ab}
\v=(T^{\ell} \a,\b) P \text{ and } (T^{\ell} \a,\b)= \v P^{-1}.
\end{equation}
Notice that if $\v$ is the $(m+n)$-stage state of $[\u_i]$, then $\v=(\a,\0) P$
and, if $\v$ is the $(m+n)$-stage state of $[\s_j]$, then $\v=(\0,\b) P$. Clearly,
\[
(T^{k} \v) P^{-1}=T^k (\v P^{-1})= (T^{\ell+k}\a,T^k \b)\mbox{ and } (\v_1+\v_2) P^{-1}=\v_1P^{-1}+\v_2P^{-1}.
\]
We view $(T^{\ell}\a,\b)$ as a state of $[L^{\ell} \u_i+ \s_j]$, keeping in mind that
the actual state is $(T^{\ell}\a,\b)P$.

Let $\p_0,\p_1,\ldots,\p_{t_1-1}$ be arbitrary $m$-stage states of $[\u_0],[\u_1],\ldots,[\u_{t_1-1}]$.
Similarly, let $\q_0,\q_1,\ldots,\q_{t_2-1}$ be arbitrary $n$-stage states of $[\s_0],[\s_1],\ldots,[\s_{t_2-1}]$. Then $(\p_i,\q_j)$ must be a state of $[L^{\ell} \u_i+ \s_j]$ for some $0 \leq \ell < \gcd(e_1,e_2)$.
Similarly, $(T^k \p_i, \q_j)$, for $0\leq k < \gcd(e_1,e_2)$, must be a state of $[L^{\ell+k}\u_i+\s_j]$,
where $\ell+k$ is reduced modulo $\gcd(e_1,e_2)$. Since the exact value of $\ell$ does not affect
the final result, we let $\ell$ be any integer. 
Thus, we obtain one state of each cycle.

We now use this new representation of the states via the mapping $P$ to construct a generic algorithm to find all conjugate pairs between any two cycles in $\Omega(f(x))$. For $C_1,C_2 \in \Omega(f(x))$, let $\v_1= (T^{x_1}\a_1,T^{x_2}\b_1)P$ be a state of $C_1$ and
$\v_2=(T^{x_3}\a_2,T^{x_4}\b_2)P$ a state of $C_2$ where $x_1, x_2, x_3, x_4 \in \Z$.
Let $e_1$ and $e_2$ be the respective period of the sequences containing states $\a_1$, $\a_2$ and $\b_1$, $\b_2$.
We assume that the period of $(\0)$ is $1$.
Algorithm~\ref{algo:b} outputs all conjugate pairs between $C_1$ and $C_2$. If $C_1=C_2$, then each conjugate pair appears twice in the output, first as $(\v,\hat{\v})$ and then as $(\hat{\v},\v)$.

\begin{algorithm}
\caption{Finding All Conjugate Pairs between Two Cycles}
\label{algo:b}
\begin{algorithmic}[1]
 \renewcommand{\algorithmicrequire}{\textbf{Input:}}
 \renewcommand{\algorithmicensure}{\textbf{Output:}}
 \Require $P, \v_1 = (\a_1,\b_1)P$, $\v_2=(\a_2,\b_2)P$, states of $C_1$ and of $C_2$, and $e_1$, $e_2$.
 \Ensure All conjugate pairs between $C_1$ and $C_2$. If $C_1=C_2$, each pair appears twice.
\Procedure{Precomputation~}{$P,\S$}\Comment{Determining $\a_3,\b_3$}
\State \textbf{return} $(\a_3,\b_3) = \S P^{-1}$
\EndProcedure
\Procedure {Subalgorithm~1~}{$\a_1$, $\a_2$, $\a_3$, $e_1$}
\State $counter_1 \gets 0$
\For {$i$ from $0$ to $e_1-1$}\Comment{If $\a_1=\0$, then $e_1=1$}
	\State $temp1 \gets \a_1+\a_3$
	\For {$i'$ from $0$ to $e_1-1$}\Comment{If $\a_2=\0$, then $e_1=1$}
		\If {$temp1=\a_2$}
			\State $counter_1 \gets counter_1+1$
			\State Store and index $(i,i')$; break from this inner loop
		\Else
			\State $temp1 \gets T(temp1)$
		\EndIf
	\EndFor
	\State $\a_1 \gets T\a_1$
\EndFor
\EndProcedure

\Procedure {Subalgorithm~2~}{$\b_1$, $\b_2$, $\b_3$, $e_2$}
\State $counter_2 \gets 0$
\For {$j$ from $0$ to $e_2-1$}\Comment{If $\b_1=\0$, then $e_2=1$}
	\State $temp2 \gets \b_1+\b_3$
	\For {$j'$ from $0$ to $e_2-1$}\Comment{If $\b_2=\0$, then $e_2=1$}
		\If {$temp2=\b_2$}
			\State $counter_2 \gets counter_2+1$
			\State Store and index $(j,j')$; break from this inner loop
		\Else
			\State $temp2 \gets T(temp2)$
		\EndIf
	\EndFor
	\State $\b_1 \gets T\b_1$
\EndFor
\EndProcedure
\Procedure {Main~}{$\v_1=(T^{x_1}\a_1,T^{x_2}\b_1)P$, $\v_2=(T^{x_3}\a_2,T^{x_4}\b_2)P$, $counter_1$, $counter2$}
\If {$counter_1=0$ or $counter_2=0$}
	\State \textbf{return} there is no conjugate pair \Comment{Propositions~\ref{prop1} and~\ref{prop2}}
\EndIf
	\For {$y$ from $1$ to $counter_1$}
    \State Take $(i,i')$ in order
		\For {$z$ from $1$ to $counter_2$}
        \State Take $(j,j')$ in order
           \If {Two elements among $\a_1$, $\a_2$, $\b_1$, $\b_2$ are $\0$}
             \State $\v \gets (T^i\a_1,T^j\b_1)P$; \textbf{output} $(\v,\hat{\v})$; \textbf{break}\Comment{Propositions~\ref{prop1} and~\ref{prop2}}
           \EndIf
	         \If {One element among $\a_1$, $\a_2$, $\b_1$, $\b_2$ is $\0$}
             \If {$\a_1=\0$ or $\b_1=\0$}
               	\If {$i'+x_3 \equiv j'+x_4 \Mod{\gcd(e_1,e_2)}$}
               		\State $\v \gets (T^i\a_1,T^j\b_1)P$; \textbf{output} $(\v,\hat{\v})$\Comment{Propositions~\ref{prop3} and~\ref{prop4}}
				\EndIf
              \Else
              	\If {$i-x_1 \equiv j-x_2 \Mod{\gcd(e_1,e_2)}$}
                	\State $\v \gets (T^i\a_1,T^j\b_1)P$; \textbf{output} $(\v,\hat{\v})$\Comment{Propositions~\ref{prop3} and~\ref{prop4}}
              	\EndIf
             \EndIf
           \EndIf
		   \If {$i-x_1 \equiv j-x_2$ and $i'+x_3 \equiv j'+x_4$ modulo  $\gcd(e_1,e_2)$}
	          \State $\v \gets (T^i\a_1,T^j\b_1)P$; \textbf{output} $(\v,\hat{\v})$\Comment{Proposition~\ref{prop5}}
	       \EndIf
       \EndFor
	\EndFor
\EndProcedure
\end{algorithmic}
\end{algorithm}

\begin{theorem}\label{th:alg}
Algorithm~\ref{algo:b} is correct.
\end{theorem}

\begin{proof}
If $counter_1=0$ or $counter_2=0$, then there does not exist a conjugate pair.

Let $(\a_3,\b_3)P=\S$ be the initial state of $[L^{c}\u_a+\s_b]$.
Without loss of generality, let $C_1=[\s_1]$ and $C_2=[\s_2]$ with initial states $\v_1=(T^{x_1}\a_1,T^{x_2}\b_1)P$ and $\v_2=(T^{x_3}\a_2,T^{x_4}\b_2)P$, respectively. If $C_1$ and $C_2$ share a conjugate pair, then there must exist an integer $0 \leq \ell < \lcm(e_1,e_2)$ such that
$L^{\ell}\s_1+\s_2$ is a shift of $L^{c}\u_a+\s_b$ or $L^{\ell}\s_1+L^{c}\u_a+\s_b$ is a shift of $\s_2$.

Suppose that none of $\a_1$, $\a_2$, $\b_1$, and $\b_2$ is $\0$. Then
there exist $0 \leq \ell, \ell' < \lcm(e_1,e_2)$ such that
$[T^{\ell}(T^{x_1}\a_1,T^{x_2}\b_1)+(\a_3,\b_3)]P
 = T^{\ell'}[(T^{x_3}\a_2,T^{x_4}\b_2)P]=[T^{\ell'}(T^{x_3}\a_2,T^{x_4}\b_2)]P$.
Hence, $T^{\ell}(T^{x_1}\a_1,T^{x_2}\b_1)+(\a_3,\b_3)=T^{\ell'}(T^{x_3}\a_2,T^{x_4}\b_2)$.
Splitting the expression into two separate components, consider
\begin{equation}\label{eqn:ab}
T^{\ell+x_1}\a_1+\a_3=T^{\ell'+x_3}\a_2 \mbox{ and } T^{\ell+x_2}\b_1+\b_3=T^{\ell'+x_4}\b_2.
\end{equation}
We have $(i,i')$ and $(j,j')$ satisfying $T^i \a_1+ \a_3 = T^{-i'}\a_2$ and $T^j\b_1+\b_3=T^{-i'}\b_2$ from the two subalgorithms.
To ensure that~(\ref{eqn:ab}) holds, it must be the case that
\[
\begin{cases}
\ell \equiv i-x_1 \Mod{e_1}\\
\ell \equiv j-x_2 \Mod{e_2}
\end{cases} \mbox{ and}
\quad
\begin{cases}
\ell' \equiv -i'-x_3 \Mod{e_1}\\
\ell' \equiv -j'-x_4 \Mod{e_2}
\end{cases}.
\]
To satisfy the requirements, we know from the Chinese Remainder Theorem that the congruences $i-x_1 \equiv j-x_2$ and
$i'+x_3 \equiv j'+x_4 $ modulo $\gcd(e_1,e_2)$ must be simultaneously satisfied.
When $(i,i')$ and $(j,j')$ satisfy the congruences,
$(T^i\a_1,T^j\b_1)P$ and $(T^{-i'}\a_2,T^{-j'}\b_2)P$ form a
conjugate pair.

If $\a_1=\0$ or $\a_2=\0$ but $\b_1, \b_2$ are not $\0$, we assume $\a_1=\0$. Hence,
$(i,i')=(0,i')$ and (\ref{eqn:ab}) becomes $\a_3=T^{\ell'+x_3}\a_2 \mbox{ and } T^{\ell+x_2}\b_1+\b_3=T^{\ell'+x_4}\b_2$. If there exists $(j,j')$ such that
$T^j\b_1+\b_3=T^{-i'}\b_2$, then there is an $\ell$ with the required properties.
It now suffices to check that $\ell'$ satisfies $\ell' \equiv -i'-x_3 \Mod{e_1}$ and
$\ell' \equiv -j'-x_4 \Mod{e_2}$ to ensure $ i'+x_3 \equiv j'+x_4  \Mod{\gcd(e_1,e_2)}$.

The other cases can be similarly proved.
\qed
\end{proof}

We gain significantly from using the new representation of the states.
Algorithm \ref{algo:b} relies on the representation to transform
the problem of finding conjugate pairs between any two cycles in $\Omega(f(x))$
into the analogous problem in the smaller sets of cycles $\Omega(p(x))$ and $\Omega(q(x))$
whose characteristic polynomials are irreducible.
The two subalgorithms ensure that the sum of the two states is equal to the indicated
part in the new representation of $\S$. Finding a conjugate
pair between any two cycles in $\Omega(f(x))$ by exhaustive search can be done in $(\lcm(e_1,e_2))^2$ times. Algorithm \ref{algo:b} requires at most $e_1^2+e_2^2$ times
to complete the same task.

In particular, if $\a_1,\a_2,\a_3$ are states of cycles in $\Omega(p(x))$
where $p(x)$ is primitive, then the connection can be made simpler by using the Zech logarithm
$\tau_{n}(\ell)$. Recall that for a primitive element $\alpha\in \F_{2^n}$, $1+\alpha^{\ell}=\alpha^{\tau_n(\ell)}$ for $1 \leq \ell <2^n-1$.
If $\a$ is an $n$-stage state of an $m$-sequence, then Lemma \ref{lem:saa} says that $\a+T^{\ell}\a=T^{\tau_n(\ell)}\a$.
Suppose it has been established that $\a:=\a_1=\a_2$ and $\a_3=T^{k}\a$. Then the
output $(i,i')$ in the first subalgorithm implies $T^i\a_1+\a_3=T^{-i'}\a_2$. Hence,
\[
T^i\a+T^k\a=T^k(\a+T^{i-k}\a)=T^{k+\tau_n(i-k)}\a=T^{-i'}\a,
\]
with $i\in \{0,1,\ldots,2^n-2\}\setminus\{k\}$. Thus, as $i$ ranges over the set $\{0,1,\ldots,2^n-2\}\setminus\{k\}$, all possible values for $(i,i')$ are given by
$\{(i,-k-\tau_n(i-k))\}$. In this special case, knowing $\tau_{n}(\ell)$ is sufficient to deduce
all possible $(i,i')$s.

\begin{remark}\label{rem:2}
Several remarks regarding Algorithm~\ref{algo:b} are in order.
\begin{enumerate}

\item The choice of a state belonging to a cycle affects neither the number of
conjugate pairs nor the states being paired in each conjugate pair.
\item The ordering of $\a_1$ and $\a_2$ matters in Subalgorithm 1. If the
output on input $(\a_1,\a_2,\a_3,e_1)$ is $(i,i')$, then that on input
$(\a_2,\a_1,\a_3,e_1)$ is $(-i',-i)$.

\item Each subalgorithm finds
"conjugate pairs" between two cycles constructed from one irreducible minimal polynomial
by exhaustive searching.
An improvement on this approach may give a significant speed up. In $\Omega(p(x))$, suppose that $\a_1$, $\a_2$, and $\a_3$ are the respective states of cycles $[\u_i]$, $[\u_j]$, and $[\u_k]$. Then Subalgorithm 1 should output $(i-k,j-k)_{t_1}$ tuples and can be stopped once all of them have been found. If it has been established that the cyclotomic number is $0$, then there is no need to
run the algorithm on this particular input case. Knowing the cyclotomic numbers allows us to truncate
the running of the algorithm. Equivalently, up to some values of $m$ and $n$, the two subalgorithms can determine the exact cyclotomic numbers computationally by using $counter_1$ and $counter_2$.

\item Let us consider the running time. The precomputation gives us $\S P^{-1}=(\a_3,\b_3)$.
Hence, we immediately infer which cycle shares a conjugate pair with $[\0]$ without having
to run the subalgorithms. By Item 2 above, the outputs of Subalgorithm 1$(\a_2,\a_1,\a_3,e_1)$
follow directly from the outputs of Subalgorithm 1$(\a_1,\a_2,\a_3,e_1)$.
Thus, Subalgorithm 1 needs to perform at most $\frac{t_1(t_1-1)}{2}+t_1=\frac{t_1(t_1+1)}{2}$
operations. The total for the two subalgorithms is therefore
$\frac{t_1(t_1+1)}{2}+\frac{t_2(t_2+1)}{2}$. The main procedure needs to be performed at most
$e_1 \cdot e_2$ times. The total number one needs to repeat Algorithm~\ref{algo:b} to complete
the adjacency matrix is bounded above by the square of the number of cycles in $\Omega(f(x))$.
\end{enumerate}
\end{remark}

To end this subsection, we provide a rough estimate on the number of de Bruijn sequences generated by our method. Let $G$ be the adjacency graph of $\Omega(p(x)q(x))$. The number is the cofactor of any entry of the symmetric and positive definite matrix $\M$ in Theorem~\ref{BEST}. With $[\0]$ as the first vertex, we use the cofactor of the entry $\M_{1,1}=1$. The product of the (remaining) entries in the main diagonal of $\M$ is a reasonably good heuristic to approximate the number.

In the main diagonal, $1$ occurs once,  $e_1$ appears $t_1$ times, $e_2$ appears $t_2$ times, and
there are $\chi:=\frac{(2^n-1)(2^m-1)}{\lcm(e_1,e_2)}=t_1 \cdot t_2 \cdot \gcd(e_1,e_2)$ other entries,
each is approximately $\lcm(e_1,e_2)$. The product of these $\chi$ entries is
\begin{equation}\label{numberdB}
E \approx \left(\lcm(e_1,e_2)\right)^{\chi}
=\left(\frac{e_1 \cdot e_2}{\gcd(e_1,e_2)}\right)^{\chi}
=\left(\frac{(2^m-1)(2^n-1)}{\chi}\right)^{\chi} \approx \left(\frac{2^{m+n}}{\chi}\right)^{\chi}.
\end{equation}
We use the last expression as a rough estimate on the number of de Bruijn sequences constructed in this work.

\section{A Detailed Example}\label{sec:ex}
This section demonstrates how the general techniques developed above fit together nicely by way of a worked-out example.
Let $p(x)=x^4+x^3+x^2+x+1$ and $q(x)=x^4+x+1$. Note that $p(x)$ is not primitive. Let $\alpha$ be a root of $q(x)$. Then $\beta=\alpha^3$ and the order of $\beta$ is $5$. Given $0 \leq j <15$, Table~\ref{table:ex} provides the representation $\left(a_{j,0},a_{j,1},a_{j,2},a_{j,3}\right)$ of $\alpha^j$ in the $\beta$ basis $\{1,\beta,\beta^2,\beta^3\}$ and $\varphi(\alpha^j)=\left(a_{j,0},a_{j+3,0},a_{j+6,0},a_{j+9,0}\right)$.

\begin{table}[h!]
\caption{List of $\varphi(\alpha^j)$ for $0 \leq j < 15$}
\label{table:ex}
\centering
\renewcommand{\arraystretch}{1.1}
\begin{tabular}{ccc|ccc|ccc}
\hline
$j$ & in $\beta$ basis & $\varphi(\alpha^j)$  & $j$ & in $\beta$ basis & $\varphi(\alpha^j)$ & $j$ & in $\beta$ basis & $\varphi(\alpha^j)$\\
\hline
$0$ & $(1,0,0,0)$ & $(1,0,0,0)$ & $5$ & $(0,0,1,1)$ & $(0,1,0,1)$ & $10$ & $(1,0,1,1)$ & $(1,1,0,1)$\\

$1$ & $(0,1,0,1)$ & $(0,1,1,1)$ & $6$ & $(0,0,1,0)$ & $(0,0,1,1)$ & $11$ & $(0,1,1,1)$ & $(0,1,0,0)$\\

$2$ & $(0,1,1,0)$ & $(0,0,1,0)$ & $7$ & $(1,0,0,1)$ & $(1,1,1,0)$ & $12$ & $(1,1,1,1)$ & $(1,1,0,0)$\\

$3$ & $(0,1,0,0)$ & $(0,0,0,1)$ & $8$ & $(1,1,1,0)$ & $(1,0,1,0)$ & $13$ & $(1,0,1,0)$ & $(1,0,1,1)$\\

$4$ & $(1,1,0,1)$ & $(1,1,1,1)$ & $9$ & $(0,0,0,1)$ & $(0,1,1,0)$ & $14$ & $(1,1,0,0)$ & $(1,0,0,1)$\\
\hline
\end{tabular}
\end{table}

By (\ref{eq:corres}), $\u_i=\left(a_{i,0},a_{i+3,0},a_{i+6,0},a_{i+9,0},a_{i+12,0} \right)$. Therefore, $\u_0=(10001)$, $\u_1=(01111)$, $\u_2=(00101)$, and $\s=(10001~00110~10111)$. We have $\Omega(p(x))=[\0]\cup [\u_0] \cup [\u_1] \cup [\u_2]$ and $\Omega(q(x))=[\0]\cup[\s]$. Thus, there are 20 disjoint cycles in $\Omega(f(x))$. Writing explicitly,
\[
\Omega(f(x))=[\0]~\cup~[\s]~\cup~\bigcup_{i=0}^2[\u_i]~\cup~\left(\bigcup_{i=0}^2\bigcup_{j=0}^{4}[L^j\u_i+\s]\right).
\]
The ordering of the $20$ cycles in use is
\[
[\0],[\u_0],[\u_1],[\u_2],[\s],[\u_0+\s],\ldots, [L^4\u_0+\s],[\u_1+\s],\ldots,[L^4\u_1+\s],[\u_2+\s],\ldots,[L^4\u_2+\s].
\]
We show how to implement Algorithm~\ref{algo:b}, work on the adjacency graph of $\Omega(f(x))$, and construct the associated matrix $\M_1$.

The $4$-stage states of $\u_0,\u_1,\u_2,$ and $\s$ are, respectively, $\p_0=(1000)$, $\p_1=(0111)$, $\p_2=(0010)$, and $\q=(1000)$.
The cycles in $\Omega(f(x))$ can be represented by their $8$-stage states
\begin{align*}
&(\0,\0)\in [\0],\ \ (\p_i,\0)\in [\u_i] \text{ for } i\in\{0,1,2\},\\
&(\0,\q)\in [\s],\ \ (T^j \p_i,\b)\in [L^j\u_i+\s] \text{ for } i\in \{0,1,2,3,4\}.
\end{align*}
Using sequences $\u_0,\u_2,\u_2$, and $\s=(10001~00110~10111)$,
\begin{equation*}
\setlength{\arraycolsep}{5pt}
P=
\begin{pmatrix}
P_1\\
P_2\\
\end{pmatrix}
=
\begin{pmatrix}
1 & 0 & 0 & 0 & 1 & 1 & 0 & 0  \\
0 & 1 & 0 & 0 & 1 & 0 & 1 & 0  \\
0 & 0 & 1 & 0 & 1 & 0 & 0 & 1  \\
0 & 0 & 0 & 1 & 1 & 0 & 0 & 0  \\
\hline
1 & 0 & 0 & 0 & 1 & 0 & 0 & 1  \\
0 & 1 & 0 & 0 & 1 & 1 & 0 & 1  \\
0 & 0 & 1 & 0 & 0 & 1 & 1 & 0  \\
0 & 0 & 0 & 1 & 0 & 0 & 1 & 1
\end{pmatrix}.
\end{equation*}
We compute $\S P^{-1}=((1101),(0101))=(T^3 \p_1,T^9 \q) =
T^9(T^4 \p_1,\q)\in[L^4 \u_1+ \s]$ to conclude that $[\0]$ and $[L^4 \u_1+ \s]$
are adjacent and the unique conjugate pair shared by $[\u_1]$ and $[\s]$ is
$(\v=(T^3 \p_1,\0)P,\hat{\v})$. In Algorithm~\ref{algo:b}, use $\a_3=(1101)=T^3 \p_1$ and $\b_3=(0101)=T^9 \q$.

Running Subalgorithm~1, we have
\begin{equation}\label{eq:out1}
\renewcommand{\arraystretch}{1.2}
\begin{array}{c|c|c|c| c|c|c}
(\a_1,\a_2)= &  (\p_0,\p_0) & (\p_0,\p_1) & (\p_0,\p_2) & (\p_1,\p_2) & (\p_2,\p_2) & (\0,\p_1)\\
\{(i,i')\}=& \{(1,1),(4,4)\} & \{(2,3),(3,1)\} & \{(0,4)\} & \{(0,3),(1,0)\} & \{(3,1),(4,2)\} & \{(0,2)\}
\end{array}.
\end{equation}
There is no output corresponding to $(\p_1,\p_1)$.
The rest of the outputs can be directly obtained by invoking Item 2 in Remark~\ref{rem:2}.
Hence, for $(\a_1,\a_2) \in \{(\p_1,\p_0),(\p_2,\p_0),(\p_2,\p_1)\}$, the respective outputs
$\{(i,i')\}$ are $\{(2,3),(4,2)\},\{(1,0)\}$, and $\{(2,0),(0,4)\}$.

On input $(\0,\q, T^9\q,15)$, Subalgorithm 2 outputs $(j,j')=(0,6)$. On $(\q,\q,T^9 \q,15)$, it outputs $\{(j,j'=-9-\tau_4(j-9))\}$ with $ j \neq 9$. The values of $\tau_4 (y)$ for $1 \leq y <15$ is reproduced here from~\cite[p.~39]{GG05}
\begin{equation*}
\renewcommand{\arraystretch}{1.2}
\setlength{\arraycolsep}{5pt}
\begin{array}{c|ccccc ccccc cccc}
y         & ~1 & 2 & 3  & 4 & 5  & 6  & 7 & 8 & 9 & 10 & 11 & 12 & 13 & 14 \\
\hline
\tau_4(y) & ~4 & 8 & 14 & 1 & 10 & 13 & 9 & 2 & 7 & 5 & 12 & 11 & 6 & 3
\end{array}.
\end{equation*}

The conjugate pair(s) shared by any two cycles in $\Omega(f(x))$ can now be determined.
We consider three cases in details.

\begin{enumerate}[label=Case \arabic*:,leftmargin=2.3\parindent]
\item The state $(\p_0,\0)P$ belongs to $C_1=[\u_0]$.

We have $(i,i') \in \{(1,1),(4,4),(2,3),(3,1),(0,4)\}$
and $(j,j')=(0,6)$. Since $\a_1 =\p_0$ and $\a_3 = T^3 \p_1$, $\a_2 \neq \0$.
Since $\b_3 \neq \0$ and $\b_1=\0$, $\b_2 \neq \0$. Hence, $C_2$ that shares at
least a conjugate pair with $[\u_0]$ must be of the form $[L^{\ell}\u_k+\s]$ for $0 \leq k <3$.
The \textsf{if loop} to consider starts from Line 45 in Algorithm~\ref{algo:b}. Note that
$x_3 = \ell$ and $x_4=0$, so $i'+\ell\equiv j' \Mod{5}$. Table~\ref{table:case1} provides the
relevant results. The state $\v$ in $C_1$ and the state $\hat{\v}$ in $C_2$ form a conjugate pair.

\begin{table}
\caption{Values obtained for Case 1}
\label{table:case1}
\centering
\renewcommand{\arraystretch}{1.2}
\begin{tabular}{c||c|c|c|c||c|c|c|c}
\hline
$k$ & $(i,i')$ & $\ell$ & $C_2$ & $\v$ & $(i,i')$ & $\ell$ & $C_2$ & $\v$\\
\hline
$0$ & $(1,1)$ & $0$ & $[\u_0+\s]$ & $(T \p_0,\0) P$ & $(4,4)$ & $2$ & $[L^{2}\u_0+\s]$ & $(T^4 \p_0,\0)P$ \\

$1$ & $(2,3)$ & $3$ & $[L^{3}\u_1+\s]$ & $(T^2 \p_0,\0)P$ & $(3,1)$  & $0$ & $[\u_1+\s]$ & $(T^3 \p_0,\0)P$ \\

$2$ & $(0,4)$ & $2$ & $[L^{2}\u_2+\s]$ & $(\p_0,\0)P$ & & & &\\
\hline
\end{tabular}
\end{table}

\item $C_1 = [L\u_0+\s]$ with $(\a_1=T \p_0,\b_1=\q)P$ as a state and $C_2 \in \{[\0],[\s],[\u_k]\}$.

Subalgorithm 1 does not output any $(i,i')$ on input $(\p_0,\0,T^3 \p_1,5)$. Thus, there is
no conjugate pair between $[L\u_0+\s]$ and either $[\0]$ or $[\s]$.

Let $[\u_k]$ have $(\a_2=\p_k, \b_2=\0)P$ as a state. On input $(\q,\0, T^9\q,15)$, $(j,j')=(-6,0)=(9,0)$. Refer to Line 51 in Algorithm~\ref{algo:b}. Since $x_1=1$ and $x_2=x_3=x_4=0$, there exists a conjugate pair between $[L\u_0+\s]$ and $[\u_k]$ if and only if
$i \equiv 0 \Mod{5}$. From (\ref{eq:out1}), this holds only if $k=2$, {\it i.e.}, $\a_2=\p_2$. Thus, only
$[L\u_0+\s]$ and $[\u_2]$ share a conjugate pair with $\v=(T^i \p_0, T^j \q)P=(\p_0, T^{9} \q)P$.

\item $C_1 = [L\u_0+\s]$ with $\v_1=(T \p_0,\q)P$ and $C_2=[L^{\ell}\u_k+\s]$ with $\v_2=(T^{\ell} \p_k,\q)P$.

It is clear that $x_1=1$, $x_2=x_4=0$, and $x_3=\ell$. Since none of $\a_1,\a_2,\b_1$, and $\b_2$ is $\0$, refer to Line 56 in the algorithm. There exists a conjugate pair between $[L \u_0+\s]$ and $[L^{\ell} \u_k+\s]$ if and only if $i-1\equiv j\Mod{5}$ and $i'+\ell\equiv j'\Mod{5}$. Table~\ref{table:case3} summarizes our computation with $j=9$ excluded from consideration.
The state $\v$ in $C_1$ and the state $\hat{\v}$ in $C_2$ form a conjugate pair.

\begin{table}[h!]
\vspace{-0.5cm}
\caption{Values obtained for Case 3}
\label{table:case3}
\centering
\renewcommand{\arraystretch}{1.2}
\begin{tabular}{c|c|c|c|c|c}
\hline
$k$ & $(i,i')$ & Requirement $\Mod{5}$ & $(j,\ell)$ & $C_2$ & $\v$ \\
\hline
$0$ & $(1,1)$ & $\ell\equiv -\tau_{4}(j-9)$ & $(0,2)$  & $[L^2 \u_0 + \s]$ & $(T \p_0, \q)P$       \\
    &         &                             & $(5,3)$  & $[L^3 \u_0 + \s]$ & $(T \p_0, T^5 \q)P$    \\
    &         &                             & $(10,1)$ & $[L \u_0 + \s]$   & $(T \p_0, T^{10} \q)P$ \\

    & $(4,4)$ & $\ell\equiv 2-\tau_4(j-9)$  & $(3,0)$  & $[\u_0 + \s]$     & $(T^4 \p_0, T^3 \q)P$    \\
    &         &                             & $(8,4)$  & $[L^4 \u_0 + \s]$ & $(T^4 \p_0, T^8 \q)P$   \\
    &         &                             & $(13,1)$ & $[L \u_0 + \s]$   & $(T^4 \p_0, T^{13} \q)P$ \\
\hline
$1$ & $(2,3)$ & $\ell\equiv 3-\tau_4(j-9)$  & $(1,4)$  & $[L^4 \u_1 + \s]$ & $(T^2 \p_0, T \q)P$     \\
    &         &                             & $(6,2)$  & $[L^2 \u_1 + \s]$ & $(T^2 \p_0, T^6 \q)P$    \\
    &         &                             & $(11,0)$ & $[\u_1 + \s]$     & $(T^2 \p_0, T^{11} \q)P$ \\

    & $(3,1)$ & $\ell\equiv -\tau_4(j-9)$   & $(2,3)$  & $[L^3 \u_1 + \s]$ & $(T^3 \p_0, T^2 \q)P$   \\
    &         &                             & $(7,4)$  & $[L^4 \u_1 + \s]$ & $(T^3 \p_0, T^7 \q)P$    \\
    &         &                             & $(12,1)$ & $[L \u_1 + \s]$   & $(T^3 \p_0, T^{12} \q)P$ \\
\hline
$2$ & $(0,4)$ & $\ell\equiv 2-\tau_4(j-9)$  & $(4,2)$  & $[L^2 \u_2 + \s]$ & $(\p_0, T^4 \q)P$       \\
    &         &                             & $(14,2)$ & $[L^2 \u_2 + \s]$ & $(\p_0, T^{14} \q)P$     \\
\hline
\end{tabular}
\end{table}
\end{enumerate}

The rest of the cases can be analyzed in a similar manner. Once all possible cases have been examined,
the completed adjacency matrix is given in (\ref{eq:AM}).

The last step is to compute the cofactor of any of the matrix's entries. Our approach constructs $2,003,859,941,621,760,000 \approx 2^{60.797}$ de Bruijn sequences. Our approximation in (\ref{numberdB}) gives $\left(\sfrac{2^{8}}{15}\right)^{15} \approx 2^{61.397}$.

\begin{equation}\label{eq:AM}
\setlength{\arraycolsep}{2pt}
\renewcommand{\arraystretch}{1.4}
\small
\M_1=
\left(
\begin{array}{@{}*{20}{r}@{}}
1 & 0 & 0 & 0 & 0 & 0 & 0 & 0 & 0 & 0 & 0 & 0 & 0 & 0 & -1 & 0 & 0 & 0 & 0 & 0\\
0 & 5 & 0 & 0 & 0 & -1 & 0 & -1 & 0 & 0 & -1 & 0 & 0 & -1 & 0 & 0 & 0 & -1 & 0 & 0\\
0 & 0 & 5 & 0 & -1 & 0 & 0 & 0 & -1 & -1 & 0 & 0 & 0 & 0 & 0 & 0 & -1 & 0 & -1 & 0\\
0 & 0 & 0 & 5 & 0 & 0 & -1 & 0 & 0 & 0 & 0 & -1 & -1 & 0 & 0 & -1 & 0 & 0 & 0 & -1\\
0 & 0 & -1 & 0 & 15 & 0 & 0 & 0 & 0 & 0 & -3 & -3 & -3 & -3 & -2 & 0 & 0 & 0 & 0 & 0\\

0 & -1 & 0 & 0 & 0 & 15 & -1 &-3& 0 & -1& 0 & -1 & -2 & -1 & -2 & -1 & 0 & 0 & -1 & -1\\
0 & 0 & 0 & -1 & 0 & -1 & 13 & -1 & -1 & -1 & -1 & -1 & -1 & -1& -2 & 0 & 0 & -2 & 0 & 0\\
0 & -1 & 0 & 0 & 0 & -3 & -1 & 15 & -1 & 0 & -1 & -2 & -1 & 0 & -2 & -1 & -1 & 0 & 0 & -1\\
0 & 0 & -1 & 0 & 0 & 0 & -1 & -1 & 13 & -2 & 0 & -1 & -1 & -3 & 0 & -1 & -1 & 0 & -1 & 0\\
0 & 0 & -1 & 0 & 0 & -1 & -1 & 0 & -2 & 13 & -3 & -1 & -1 & 0 & 0 & 0 & -1 & 0 & -1 & -1\\

0 & -1 & 0 & 0 & -3 & 0 & -1 & -1 & 0 & -3 & 15 & 0 & 0 & 0 & 0 & -2 & -1 & -1 & -1 & -1\\
0 & 0 & 0 & -1 & -3 & -1 & -1 & -2 & -1 & -1 & 0 & 15 & 0 & 0 & 0 & 0 & 0 & -1 & -3 & -1\\
0 & 0 & 0 & -1 & -3 & -2 & -1 & -1 & -1 & -1 & 0 & 0 & 15 & 0 & 0 & -1 & -3 & -1 & 0 & 0\\
0 & -1 & 0 & 0 & -3 & -1 & -1 & 0 & -3 & 0 & 0 & 0 & 0 & 15 & 0 & -1 & -1 & -1 & -1 & -2\\
-1 & 0 & 0 & 0 & -2 & -2 & -2 & -2 & 0 & 0 & 0 & 0 & 0 & 0 & 15 & -2 & 0 & -2 & 0 & -2\\
0 & 0 & 0 &  -1 & 0 & -1 & 0 & -1 & -1 & 0 & -2 & 0 & -1 & -1 & -2 & 15 & 0 & -1 & -1 & -3\\

0 & 0 & -1 & 0 & 0 & 0 & 0 & -1 & -1 & -1 & -1 & 0 & -3 & -1 & 0 & 0 & 13 & -1 & -2 & -1\\
0 & -1 & 0 & 0 & 0 & 0 & -2 & 0 & 0 & 0 & -1 & -1 & -1 & -1 & -2 & -1 & -1 & 13 & -1 & -1\\
0 & 0 & -1 & 0 & 0 & -1 & 0 & 0 & -1 & -1 & -1 & -3 & 0 & -1 & 0 & -1 & -2 & -1 & 13 & 0\\
0 & 0 & 0 & -1 & 0 & -1 & 0 & -1 & 0 & -1 & -1 & -1 & 0 & -2 & -2 & -3 & -1 & -1 & 0 & 15
\end{array}
\right).
\end{equation}

\section{Some Special Cases}\label{sec:spcases}
Theorem~\ref{main} makes clear that, for general irreducible polynomials $p(x)$ and $q(x)$, determining all conjugate pairs between any two cycles can be quite complicated. This section highlights three special cases for which the process is much simpler.
\subsection{The orders of $p(x)$ and $q(x)$ are relatively prime}\label{subsec1}
Based on Lemma~\ref{lem:cycle-f}, when $\gcd(e_1,e_2)=1$,
\begin{equation}\label{equ:cycle-f1}
\Omega(f(x))=[\0]~\cup~\bigcup_{i=0}^{t_1-1}[\u_i]~\cup~\bigcup_{j=0}^{t_2-1}[\s_j]~\cup
             ~\left(\bigcup_{i=0}^{t_1-1}~\bigcup_{j=0}^{t_2-1}[\u_i+\s_j]\right).
\end{equation}
Directly applying Propositions~\ref{prop1} to~\ref{prop5} leads to the next result.
\begin{proposition}\label{prop:rp}
Let $\S \in [\u_a+\s_b]$ for some $a$ and $b$. The following properties hold.
\begin{enumerate}
\item $[\0]$ and $[\u_a+\s_b]$ are adjacent.
\item Let $0 \leq i <t_1$ and $0 \leq j <t_2$. There is no conjugate pair between
$[\u_i]$ and $[\u_j]$ and between $[\s_i]$ and $[\s_j]$. There is a conjugate pair
shared by $[\u_i]$ and $[\s_j]$ if and only if $i=a$ and $j=b$, in which case
the pair is unique.
\item There is no conjugate pair between $[\u_i]$ and $[\u_j+\s_k]$ if $k \neq b$.
For $0 \leq i,j <t_1$, the number of conjugate pairs between
$[\u_i]$ and $[\u_j+\s_b]$ is the cyclotomic number $(i-j,a-j)_{t_1}$.
\item There is no conjugate pair between $[\s_i]$ and $[\u_k+\s_j]$ if $k \neq a$.
For $0\leq i,j<t_2$, the number of conjugate pairs between
$[\s_i]$ and $[\u_a+\s_j]$ is the cyclotomic number $(i-j,b-j)_{t_2}$.
\item For $0\leq i_1,i_2<t_1$ and $0\leq j_1,j_2<t_2$, the number of conjugate pairs
between two {\bf distinct} cycles $[\u_{i_1}+\s_{j_1}]$ and $[\u_{i_2}+\s_{j_2}]$ is
$(i_2-i_1,a-i_1)_{t_1}\cdot(j_2-j_1,b-j_1)_{t_2}$.
\item The number of conjugate pairs between $[\u_{i}+\s_{j}]$ and itself is
$\frac{1}{2} (0,a-i)_{t_1} \cdot (0,b-j)_{t_2}$.
\end{enumerate}
\end{proposition}

\subsection{Both $p(x)$ and $q(x)$ are primitive polynomials}\label{subsec2}
Let $p(x)$ and $q(x)$ be distinct primitive polynomials. Then,
$t_1=t_2=1$ and $r:=\gcd(e_1,e_2)=\gcd(2^m-1,2^n-1)=2^{\gcd(m,n)}-1$. Consulting~(\ref{equ:cycle-f}),
\begin{equation}\label{equ:cycle-fs}
\Omega(f(x))=[\0]~\cup~[\u]~\cup~[\s]~\cup~\bigcup_{i=0}^{r-1}[L^i\u+\s].
\end{equation}
\begin{proposition}\label{prop:pp}
Let $a$ be such that $\S \in [L^a\u+\s]$ and $r=\gcd(e_1,e_2)=2^{\gcd(m,n)}-1$. Then
\begin{enumerate}
\item $[\0]$ and $[L^a\u+\s]$ are adjacent.
\item $[\u]$ and $[\s]$ share a unique conjugate pair.
\item There are $\sfrac{e_1}{r}-1$ conjugate pairs between $[\u]$ and $[L^a\u+\s]$
and $\sfrac{e_1}{r}$ conjugate pairs between $[\u]$ and $[L^i\u+\s]$ when $0 \leq i < r$ and $i \neq a$.
\item There are $\sfrac{e_2}{r}-1$ conjugate pairs between $[\s]$ and $[L^a \u+\s]$
and $\sfrac{e_2}{r}$ conjugate pairs between $[\s]$ and $[L^i\u+\s]$ when $0 \leq i< r$ and $i \neq a$.
\item Let $0 \leq i \neq j < r$. The number of conjugate pairs between $[L^i\u+\s]$ and $[L^j\u+\s]$ is
\begin{align}\label{equ:numus}
N(i,j)&=N(j,i):=\left|\{0\leq k<\lcm(e_1,e_2)\}\right| \text{ where} \notag\\
\tau_n(k) &\equiv\tau_m(k+i-j)+j-a \Mod{r},~ k \not\equiv 0 \Mod{e_1} \text{, and } k \not\equiv j-i \Mod{e_2}.
\end{align}
If $i=j$, we halve the count in~(\ref{equ:numus}).
\end{enumerate}
\end{proposition}

\begin{proof}
Since $\S \in [L^a\u+\s]$, $[\0]$ and $[L^a\u+\s]$ are adjacent. So are $[\u]$ and $[\s]$,
which share a unique conjugate pair.

Consider $[\u]$ and $[L^i\u+\s]$ for $0 \leq i < r$.
By Lemma~\ref{lem:saa}, for $0 \leq k <e_1$,
\[
L^k \u+ L^i \u+\s = L^i (L^{k-i}\u+\u)+\s = L^i L^{\tau_m(k-i)}\u+\s = L^{i+\tau_m(k-i)}\u+\s
\]
is shift equivalent to $L^a\u+\s$ if and only if
\begin{equation}\label{equ:moduv}
i+\tau_m (k-i)\equiv a \Mod{r} \iff \tau_m(k-i)\equiv a-i \Mod{r}.
\end{equation}
Since $\tau_m$ is a permutation,~(\ref{equ:moduv}) has $\sfrac{e_1}{r}-1$ solutions when $i=a$ and $\sfrac{e_1}{r}$ solutions when $i \neq a$.

Consider $[\s]$ and $[L^i\u+\s]$ for $0 \leq i < r$.
For $0\leq k<e_2$, Lemma~\ref{lem:saa} says that
\[
L^k\s+L^i\u+\s = L^i\u+\s+L^k\s = L^i\u+L^{\tau_n(k)}\s =L^{\tau_n(k)}(L^{i-\tau_n(k)}\u+\s)
\]
is shift equivalent to $L^a\u+\s$ if and only if
\begin{equation}\label{equ:modsv}
i-\tau_n(k)\equiv a \Mod{r} \iff \tau_n(k)\equiv i-a \Mod{r}.
\end{equation}
Thus, (\ref{equ:modsv}) has $\sfrac{e_2}{r}-1$ solutions for $i=a$ and $\sfrac{e_2}{r}$ solutions
for $i \neq a$.

For the last assertion, we count the number of conjugate pairs between
$[L^i\u+\s]$ and $[L^j\u+\s]$ for $0 \leq i,j <r$. By Lemma~\ref{lem:saa}, for $0 \leq k
< \lcm(e_1,e_2)$,
\begin{align*}
L^k(L^i\u+\s) + L^j\u+\s &= L^{k+i}\u+L^j\u+L^k\s+\s
= L^j(L^{k+i-j}\u+\u)+L^{\tau_n(k)}\s\\
&= L^{\tau_n(k)}(L^{j-\tau_n(k)}L^{\tau_m(k+i-j)}\u+\s)
= L^{\tau_n(k)}(L^{j-\tau_n(k)+\tau_m(k+i-j)}\u+\s)
\end{align*}
is shift equivalent to $L^a\u+\s$ if and only if $j-\tau_n(k)+\tau_m(k+i-j) \equiv a \Mod{r}$.
Equivalently, the condition can be written as
\begin{equation}\label{equ:modssv}
\tau_n(k) \equiv \tau_m(k+i-j)+j-a \Mod{r}.
\end{equation}
Thus, if $i\neq j$, the number of conjugate pairs is indeed given by (\ref{equ:numus}), and we halve the number when $i=j$.\qed
\end{proof}

When $\gcd(e_1,e_2)=1$, Item 5 in Proposition~\ref{prop:pp} is a special case of~\cite[Theorem 2]{Li14-2}. We present it here as a corollary.
\begin{corollary}\label{cor1}
Adding $\gcd(e_1,e_2)=1$ to the assumptions of Proposition~\ref{prop:pp},
the number of conjugate pairs between $[\u+\s]$ and itself is
\begin{equation}\label{equ:numusgcd=1}
\sfrac{1}{2} \cdot \left| \{0\leq k<e_1e_2~|~k \not\equiv 0 \Mod{e_1};~
k \not \equiv 0 \Mod{e_2}\}\right|
=\sfrac{1}{2} \cdot (e_1-1)(e_2-1).
\end{equation}
\end{corollary}

If $e_1 \mid e_2$, that is when $\gcd(e_1,e_2) = e_1$, we can derive an explicit formula.
\begin{corollary}\label{cor2}
If $e_1 \mid e_2$, the number of conjugate pairs between $[L^i\u+\s]$ and $[L^j\u+\s]$, for $0\leq i \neq j < e_1$, is
\begin{equation}\label{equ:numusgcde}
N(i,j)=\sum_{\substack{\ell=0 \\ \ell \not\equiv j-i \Mod{e_1}}}^{e_1-1}(\ell,\tau_m(\ell + i-j)+j-a)_{e_1}.
\end{equation}
If $i=j$, we halve the count in~(\ref{equ:numusgcde}).
\end{corollary}
\begin{proof}
Since $\gcd(e_1,e_2)=e_1$, rewrite~(\ref{equ:numus}) as
\begin{align}\label{equ:modssv1}
N(i,j)&=\left|\{0 < k < e_2 \}\right| \text{ satisfying } k \not\equiv j-i \Mod{e_1} \notag\\
\text{ and } \tau_n(k) &\equiv \tau_m(k+i-j)+j-a \Mod{e_1}.
\end{align}
More explicitly, we compute for
\begin{equation}\label{equ:modssv2}
\sum_{\substack{\ell=0,\\ \ell \neq j-i}}^{e_1-1} \left|\{\ell+t \cdot e_1 \}\right| \text{ with }
0\leq t<\frac{e_2}{e_1} \text{ and } \tau_n(\ell+t \cdot e_1) \equiv\tau_m(\ell+i-j)+j-a \Mod{e_1}.
\end{equation}
Based on the equivalence of (\ref{eq:cycnum}) and (\ref{eq:cycnum2}), we confirm that (\ref{equ:modssv2}) and~(\ref{equ:numusgcde}) are the same.\qed
\end{proof}
\subsection{De Bruijn Sequences of order $n+2$}\label{subsec:3}
Let $p(x)$ be a primitive polynomial of degree $n > 2$. We look into the construction
from LFSRs with characteristic polynomial $(x^2+x+1)p(x)$. The exact number of de
Bruijn sequences constructed can be determined.

It is clear that $\Omega(p(x))=[\0]\cup[\s]$ and $\Omega(x^2+x+1)=[\0]\cup[\u]$,
where $\s$ and $\u$ are maximal length sequences with period $2^n-1$ and $3$, respectively.
In fact, $\u$ must be a shift of $(110)$.
By Lemma~\ref{lem:cycle-f} and the fact that $\gcd(3,2^n-1)$ is $1$ if $n$ is odd and is $3$ if $n$ is even,
\begin{equation*}
\Omega(f(x))=
\begin{cases}
[\0]\cup[\u]\cup[\s]\cup[\u+\s] &\text{ if } n \text{ is odd,}\\
[\0]\cup[\u]\cup[\s]\cup\bigcup_{i=0}^2 [L^i\u+\s] &\text{ if } n \text{ is even}.
\end{cases}
\end{equation*}

The next proposition follows from Proposition~\ref{prop:pp} and Corollary~\ref{cor1}.
\begin{proposition}\label{prop:c3odd}
Let $n \geq 3$ be odd. Figure~\ref{fig1} shows the adjacency graph, based on the following facts.
\begin{enumerate}
\item There is a unique conjugate pair each between $[\0]$ and $[\u+\s]$ and between $[\u]$ and $[\s]$.
\item $[\u]$ and $[\u+\s]$ share $2$ conjugate pairs.
\item $[\s]$ and $[\u + \s]$ share $2^n-2$ conjugate pairs.
\item $[\u+\s]$ shares $2^n-2$ conjugate pairs with itself.
\end{enumerate}
\begin{figure}[h]
\vspace{-0.6cm}
\begin{tikzpicture}[auto,node distance=2.2cm,every loop/.style={},
                    thick,main node/.style={draw,font=\sffamily\bfseries}]

  \node[main node] (1) {$[\u]$};
  \node[main node] (2) [right of=1] {$[\s]$};
  \node[main node] (3) [right of=2] {$[\u+\s]$};
  \node[main node] (4) [right of=3] {$[\0]$};

  \path[every node/.style={font=\sffamily}]
    (1) edge node [above] {$1$} (2)
        edge [bend left] node [above] {$2$} (3)
    (3) edge node [above] {$1$} (4)
        edge node [above] {$2^n-2$} (2)
        edge [loop above] node {$2^n-2$} (3);
\end{tikzpicture}
\centering
\caption{The adjacency graph of $\Omega((x^2+x+1)~p(x))$ for odd $n \geq 3$.}
\label{fig1}
\end{figure}
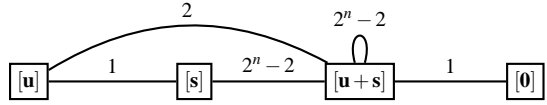
\end{proposition}

When $n$ is even, we need some results on cyclotomic numbers.
\begin{lemma}\label{lemma:cn}(See~\cite{Storer67} or~\cite[Section~4]{HH96}) Let $n$ be even. Then
\begin{align*}
A :&=(0,0)_3 = \frac{1}{9} \cdot \left(2^n+(-2)^{\frac{n}{2}+1}-8\right)\text{, } C :=(1,2)_3 = \frac{1}{9} \cdot \left(2^n+(-2)^{\frac{n}{2}+1}+1 \right)\text{, and}\\
B :&=(0,1)_3 =(1,1)_3 = (0,2)_3 =(2,2)_3= \frac{1}{9} \cdot \left(2^n+(-2)^{\frac{n}{2}}-2\right).
\end{align*}
For $i > j$, we have $(i,j)_3=(j,i)_3$.
\end{lemma}
The next result follows from Proposition~\ref{prop:pp} and Corollary~\ref{cor2}.
\begin{proposition}\label{prop:c3even}
Let $n\geq 4$ be even. Without loss of generality, suppose that $\S \in [\u + \s]$. Then
\begin{enumerate}
\item $[\0]$ and $[\u+\s]$ share a unique conjugate pair.
\item $[\u]$ shares a unique conjugate pair each with $[\s]$, $[L \u+ \s]$ and $[L^2 \u + \s]$.\\
There is no conjugate pair between $[\u]$ and $[\u+\s]$.
\item $[\s]$ and $[\u+\s]$ share $\frac{2^n-1}{3}-1$ conjugate pairs.
\item For $\ell \in \{1,2\}$, $[\s]$ and $[L^{\ell} \u+\s]$ share $\frac{2^n-1}{3}$ conjugate pairs.
\item Let $N(i,j)=N(j,i)$, for $0 \leq i,j<3$, be the number of conjugate pairs between
$[L^{i} \u +\s]$ and $[L^{j} \u + \s]$. Based on Lemma~\ref{lemma:cn}, $N(0,0) = C$,
$N(0,1) = 2B$, $N(0,2)=2B$, $N(1,1) = B$, $N(1,2) = A+C$, and $N(2,2)=B$.
\end{enumerate}
The adjacency graph is shown in Figure~\ref{fig2} with $Z:=\frac{2^n-1}{3}$.
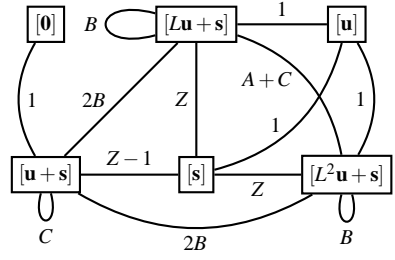
\begin{figure}[h]
\vspace{-0.1cm}
\begin{tikzpicture}[auto, node distance=2.0cm, every loop/.style={},
                    thick,main node/.style={draw,font=\sffamily\bfseries}]

  \node[main node] (1) {$[\0]$};
  \node[main node] (2) [below of=1] {$[\u+\s]$};
  \node[main node] (4) [right of=1] {$[L\u+\s]$};
  \node[main node] (3) [below of=4] {$[\s]$};
  \node[main node] (5) [right of=4] {$[\u]$};
  \node[main node] (6) [below of=5] {$[L^{2}\u+\s]$};

  \path[every node/.style={font=\sffamily}]
    (1) edge [bend right] node [right] {$1$} (2)
    (2) edge node [above] {$Z-1$} (3)
        edge node [left] {$2B~$} (4)
        edge [loop below] node {$C$} (2)
        edge [bend right] node [below] {$2B$} (6)
    (4) edge [loop left] node {$B$} (4)
    		edge node [above] {$1$} (5)
        edge node [left] {$Z$} (3)
        edge [bend left] node [left] {$A+C~$} (6)
    (3) edge [bend right] node [left] {$1~$} (5)
        edge node [below] {$Z$} (6)
    (5) edge [bend left] node [left] {$1$} (6)
    (6) edge [loop below] node {$B$} (6);

\end{tikzpicture}
\centering
\caption{The adjacency graph of $\Omega((x^2+x+1)~p(x))$ for even $n \geq 4$.}
\label{fig2}
\end{figure}
\end{proposition}
\begin{proof}
The first four items follow directly from Proposition~\ref{prop:pp}.
The last item is deduced from (\ref{equ:numusgcde}) using $\tau_2(1)=2$ and $\tau_2(2)=1$.
\begin{align*}
N(0,0)&=\frac{1}{2}\sum_{\ell=1}^{3-1}(\ell,\tau_m(\ell + 0-0)+0-0)_{3}=\frac{1}{2}((1,2)_3+(2,1)_3)=(1,2)_3=C,\\
N(0,1)&=N(1,0)=(0,\tau_2(0-1)+1)_3+(2,\tau_2(2-1)+1)_3=(0,2)_3+(2,0)_3=2B.
\end{align*}
The other values can be obtained in a similar way.\qed
\end{proof}
\begin{theorem}\label{thm:count}
Let $\A_n$ be the set of all primitive polynomial of degree $n>2$ over $\F_2$. Let $p(x) \in \A_n$. The total number of de Bruijn sequences constructed from LFSRs with characteristic polynomials $(x^2+x+1)~p(x)$ is
\[
\begin{cases}
(3\cdot2^n-4) \cdot \frac{\phi(2^n-1)}{n} &\text{ if } n \geq 3 \text{ is odd}\\
\left[2^{3n}-\frac{9\cdot2^{2n+4}-(-2)^{\sfrac{3n}{2}+4}-3\cdot2^{n+6}+2^6}{27}\right]\frac{\phi(2^n-1)}{n}
&\text{ if } n \geq 4 \text{ is even.}
\end{cases}.
\]
\end{theorem}
\begin{proof}
We count the number of spanning trees in the adjacency graph.

Let $n \geq 3$ be odd. Label the vertices $v_1=[\0]$, $v_2=[\u]$, $v_3=[\s]$, and $v_4=[\u+\s]$ to derive
\begin{equation*}
\setlength{\arraycolsep}{5pt}
\M=
\begin{pmatrix}
    1 & 0  & 0        & -1 \\
    0 & 3  & -1       & -2 \\
    0 & -1 & 2^n-1    & 2-2^n \\
    -1 & -2 & 2-2^n  & 2^n+1
\end{pmatrix}
\end{equation*}
with cofactor $\M(3,3)= 3 \cdot 2^n-4$.

Let $n \geq 4$ be even. Label the vertices $v_1=[\0]$, $v_2=[\u]$, $v_3=[\s]$, $v_4=[\u+\s]$, $v_5=[L \u + \s]$, and $v_6=[L^2 \u+\s]$ to derive
\begin{equation*}
\renewcommand{\arraystretch}{1.5}
\setlength{\arraycolsep}{6pt}
\M'=
\begin{pmatrix}
    1  & 0  & 0        & -1       & 0       & 0   \\
    0  & 3  & -1       & 0        & -1      & -1  \\
    0  & -1 & 2^n-1    & \frac{4-2^n}{3}    & \frac{1-2^n}{3} & \frac{1-2^n}{3} \\
   -1  & 0  & \frac{4-2^n}{3} & 4B+\frac{2^n-1}{3} & -2B   & -2B \\
    0  & -1 & \frac{1-2^n}{3} & -2B                & A+2B+C+\frac{2^n+2}{3} & -(A+C) \\
    0  & -1 & \frac{1-2^n}{3} & -2B                & -(A+C)    & A+2B+C+\frac{2^n+2}{3}
\end{pmatrix}
\end{equation*}
with cofactor $\displaystyle{\M'(5,5)=2^{3n}-\frac{1}{27} \left(9\cdot2^{2n+4}-(-2)^{\frac{3n}{2}+4}-3\cdot2^{n+6}+2^6\right)}$.

By~\cite[Theorem 5]{Jansen91}, applying the cycle joining method to two distinct LFSRs results in distinct de Bruijn sequences. Since there are $\sfrac{\phi(2^n-1)}{n}$ choices for the primitive polynomial $p(x)$ (see {\it e.g.}~\cite[page~70]{GG05}), the desired conclusion follows.
\qed
\end{proof}

Table~\ref{table:nplus2} provides the number for $3 \leq n \leq 10$ based on Theorem~\ref{thm:count}.

\begin{table}[h!]
\vspace{-0.5cm}
\caption{Number of de Bruijn sequences of order $n+2$ in Theorem~\ref{thm:count} for $3 \leq n \leq 10$.}
\label{table:nplus2}
\centering
\renewcommand{\arraystretch}{1.2}
\begin{tabular}{c|c|c|c|c|c|c|c|c}
\hline
$\deg(p(x))=n$  & $3$ & $4$ & $5$ & $6$ & $7$ & $8$ & $9$ & $10$ \\
Order $=n+2$     & $5$ & $6$ & $7$ & $8$ & $9$ & $10$ & $11$ & $12$ \\
\# per $p(x)$    & $20$ & $2,880$ & $92$ & $240,448$ & $380$ & $16,431,936$ & $1,532$ & $1,068,137,280$ \\
$|\A_n|$        & $2$ & $2$ & $6$ & $6$ & $18$ & $16$ & $48$ & $60$ \\
\hline
\end{tabular}
\end{table}

\begin{remark}
One can also derive Theorem~\ref{thm:count} by applying~\cite[Proposition 5]{Li16} on relevant
results in \cite{Li14-2}. The latter reference uses $(1+x^3)p(x)$ as the characteristic polynomial. Hence, proper modifications are needed before the count can be established. In our present work, properties of cyclotomic numbers play a crucial role in establishing the count directly.
\end{remark}

\begin{example}
Let $n=3$ and $p(x)=x^3+x+1$, making $f(x)=x^5+x^4+1$. This produces $20$ distinct $32$-periodic de Bruijn sequences.
\[
\Omega(f(x))=\{[\0],[\u=(110)],[\s=(0010111)],[\u + \s = (1111010~1001100~0100001)]\}.
\]
Cycles $[\0]$ and $[\u + \s]$ share the pair $(\vX_1=(00000),\widehat{\vX}_1)$.
Cycles $[\u]$ and $[\s]$ are adjacent with a shared pair $(\vX_2=(11011),\widehat{\vX}_2)$.
Cycles $[\u]$ and $[\u + \s]$ share $2$ conjugate pairs, namely
$(\vX_3=(10110),\widehat{\vX}_3)$ and $(\vX_4=(01101),\widehat{\vX}_4)$. To derive one de Bruijn sequence, select the spanning tree
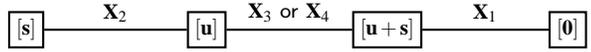
\begin{figure}[h]
\vspace{-0.7cm}
\begin{tikzpicture}[auto,node distance=2.4cm,every loop/.style={},
                    thick,main node/.style={draw,font=\sffamily\bfseries}]

  \node[main node] (1) {$[\s]$};
  \node[main node] (2) [right of=1] {$[\u]$};
  \node[main node] (3) [right of=2] {$[\u+\s]$};
  \node[main node] (4) [right of=3] {$[\0]$};

  \path[every node/.style={font=\sffamily}]
    (1) edge node [above] {$\vX_2$} (2)
    (2) edge node [above] {$\vX_3$ or $\vX_4$} (3)
    (3) edge node [above] {$\vX_1$} (4);
\end{tikzpicture}
\centering
\label{fig3}.
\end{figure}

Applying the CJ method on $[\0]$ and $[\u+\s]$ using the conjugate pairs defined by
$\vX_1$ yields $(\colorbox[gray]{0.8}{1\underline{0000}}\underline{0}1111101010011000)$; on $[\u]$ and $[\s]$
using $\vX_2$ results in
$(\colorbox[gray]{0.8}{11011}10\colorbox[gray]{0.8}{01011}0)$. We now choose
$\vX_3$ to combine the two larger cycles to get the de Bruijn sequence
\[
(00000 11111 0101\colorbox[gray]{0.8}{00110}11100\colorbox[gray]{0.8}{10110}001)\]
whose feedback function is $h(x_0,x_1,x_2,x_3,x_4)=$
\begin{align*}
&x_0+x_4+(x_1+1)(x_2+1)(x_3+1)(x_4+1)+x_1(x_2+1)x_3x_4+(x_1+1)x_2x_3(x_4+1)=\\
&x_1x_2x_3x_4+x_1x_2x_4+x_1x_2+x_1x_3+x_1x_4+x_2x_4+x_3x_4+x_0+x_1+x_2+x_3+1.
\end{align*}
One can opt to use $\vX_4$ instead of $\vX_3$. The derivation is an easy exercise for the reader.
\end{example}

\begin{example}
Let $n=4$ and $p(x)=x^4+x+1$. Hence, $f(x)=x^6+x^5+x^4+x^3+1$, from which $2880$ distinct de Bruijn sequences with period $64$ can be constructed.
\begin{align*}
\Omega(f(x))= \{&[\0],[\u = (110)], [\s = (00010~01101~01111)],
[\u+\s = (11001~00000~11001)],\\
&[L \u+\s  = (10100~10110~00010)],[L^2 \u+\s = (01111~11011~10100)]\}.
\end{align*}
Cycles $[\0]$ and $[\u + \s]$ share the pair $(\vX_1=(000000),\widehat{\vX}_1)$.
Cycles $[\u]$ and $[\s]$ share the pair $(\vX_2 =(101101),\widehat{\vX}_2)$.
The pair $(\vX_3=(110110),\widehat{\vX}_3)$ is shared by $[\u]$ and $[L \u + \s]$. Cycles
$[\u]$ and $[L^2 \u+ \s]$ share the pair $(\vX_4=(011011),\widehat{\vX}_4)$. Finally,
the pair $(\vX_5=(100110),\widehat{\vX}_5)$ is between $[\s]$ and $[\u+ \s]$. To construct one de Bruijn sequence, use the spanning tree
\begin{figure}[h]
\begin{tikzpicture}[auto,node distance=1.3cm,every loop/.style={},
                    thick,main node/.style={draw,font=\sffamily\bfseries}]

  \node[main node] (1) {$[\0]$};
  \node[main node] (2) [below of=1] {$[\u+\s]$};
  \node[main node] (3) [right of=2] {$[\s]$};
  \node[main node] (4) [right of=3] {$[\u]$};
  \node[main node] (5) [right of=4] {$[L^{2}\u + \s]$};
  \node[main node] (6) [above of=4] {$[L\u + \s]$};

  \path[every node/.style={font=\sffamily}]
    (1) edge node [right] {$X_1$} (2)
    (2) edge node [above] {$X_5$} (3)
    (3) edge node [above] {$X_2$} (4)
    (4) edge node [above] {$X_4$} (5)
    (4) edge node [left]  {$X_3$} (6);
\end{tikzpicture}
\centering
\label{fig4}.
\end{figure}

Table~\ref{table:ex3} lists down the joined cycles using the pairs defined by,
in order, $\vX_1$, $\vX_5$, $\vX_2$, $\vX_4$, and $\vX_3$. The one
in the last row is de Bruijn with feedback function $h(x_0,x_1,x_2,x_3,x_4,x_5)=$
\begin{multline*}
x_0+x_3+x_4+x_5+(x_1+1)(x_2+1)(x_3+1)(x_4+1)(x_5+1)+(x_1+1)x_2x_3(x_4+1)x_5+\\
x_1(x_2+1)x_3x_4(x_5+1)+ x_1x_2(x_3+1)x_4x_5+(x_1+1)(x_2+1)x_3x_4(x_5+1).
\end{multline*}

\begin{table}[h]
\vspace{-0.6cm}
\caption{Applying the cycle joining method on the spanning tree.}
\label{table:ex3}
\centering
\renewcommand{\arraystretch}{1.2}
\begin{tabular}{cl}
\hline
Link & Resulting Cycle \\
\hline
$\vX_1$ & $(\underline{1}~\colorbox[gray]{0.8}{\underline{000~00}0}1~1001~1100)$ \\
$\vX_5$ & $(1000~\colorbox[gray]{0.8}{0001~10}10~1111~000\colorbox[gray]{0.8}{1~0011~0}011~100)$ \\
$\vX_2$ & $(1000~0\colorbox[gray]{0.8}{001~\underline{101}}\underline{1~01}01~1110~0010~0110~0111~00)$ \\
$\vX_4$ & $(1000~00\colorbox[gray]{0.8}{01~1011}1010~0011~1\colorbox[gray]{0.8}{111~011}0~1011~1100~0100~1100~1110~0)$ \\
$\vX_3$ & $(1000~0001~1011~1010~0011~11\colorbox[gray]{0.8}{11~0110}~0001~0101~0\colorbox[gray]{0.8}{010~110}1~0111~1000~1001~1001~1100)$ \\
\hline
\end{tabular}
\end{table}
\end{example}

\section{More General Characteristic Polynomials}\label{sec:genpoly}
This section briefly touches upon the construction of de Bruijn sequences based on LFSRs with characteristic polynomials other than those discussed above.

When the characteristic polynomial takes a certain form, the adjacency graph contains no loops (see, {\it e.g.},~\cite[Proposition~2]{Li14-2}). The same holds for a much larger class of polynomials. Since $1+x^3 = (1+x)(1+x+x^2)$,~\cite[Proposition~2]{Li14-2} is subsumed by the next result.
\begin{proposition}\label{prop:nopair}
Let $1+x,p_1(x),p_2(x),\ldots,p_s(x) \in \F_2[x]$ be pairwise distinct irreducible polynomials and $f(x)=(1+x)~\prod_{i=1}^{s} p_i(x)$.
The adjacency graph of $\Omega(f(x))$ contains no loops.
\end{proposition}
\begin{proof}
Let $C$ be a cycle in $\Omega(f(x))$ that shares a conjugate pair with itself. Then the minimal polynomial of $C$ must be $f(x)$.
Hence, $C=[\1+ L^{i_1} \u_1+ \ldots + L^{i_{s-1}} \u_{s-1}+ \u_s]$ for some integers $i_1,i_2,\ldots,i_{s-1}$ with $p_i(x)$ being the minimal polynomial of $\u_i$ for $1 \leq i \leq s$. Thus, for some $\ell \in \Z$, we get $(\1+ L^{i_1} \u_1+ \cdots + L^{i_{s-1}} \u_{s-1}+ \u_s) +
L^{\ell} (\1+ L^{i_1} \u_1 + \cdots + L^{i_{s-1}} \u_{s-1}+ \u_s)
= L^{i'_1} \u'_1+ \cdots + L^{i'_{s-1}} \u'_{s-1} + L^{i'_s} \u'_s$, where the characteristic polynomial of $\u'_i$ is $p_i(x)$.
Now, the degree of the minimal polynomial of the resulting sequence must be $< \deg(f(x))$. Thus, it cannot contain $\S$.\qed
\end{proof}

Consider the characteristic polynomial $h(x)=(1+x)f(x)$ with $f(x)$ given in Lemma~\ref{lem:cycle-f}. The only nonzero sequence having $1+x$ as its characteristic polynomial is $\1$. The cycle structure of $\Omega(h(x))$ follows directly from Lemmas~\ref{lemma1} and~\ref{lem:cycle-f}.
\begin{lemma}\label{lem:cycle-g}
The cycle structure of $\Omega(h(x))$ is
\begin{align*}
&[\0]~\cup~[\1]~\cup~\bigcup_{i=0}^{t_1-1}[\u_i]~\cup~\bigcup_{i=0}^{t_1-1}[\1+\u_i]~\cup~\bigcup_{j=0}^{t_2-1}[\s_j] ~\cup~\bigcup_{j=0}^{t_2-1}[\1+\s_j]~\cup\\
&\left(\bigcup_{i=0}^{t_1-1}~\bigcup_{j=0}^{t_2-1}~\bigcup_{k=0}^{~\gcd(e_1,e_2)-1~}[L^k\u_i+\s_j]\right)~\cup~\left(\bigcup_{i=0}^{t_1-1}~\bigcup_{j=0}^{t_2-1}~\bigcup_{k=0}^{~\gcd(e_1,e_2)-1~}[\1+L^k\u_i+\s_j]\right).
\end{align*}
\end{lemma}

Any cycle in $\Omega(h(x))$ can be described as $[a_0 \1 + a_1 L^k\u_i + a_2 \s_j]$ with $i,j,k\in\Z$
and $a_0,a_1,a_2\in\F_2$. Lemma~\ref{lem:cycle-g} leads us directly to the next result.
\begin{proposition}\label{prop:pair}
If $[a_0 \1 + a_1 L^k\u_i + a_2 \s_j]$ and $[a'_0 \1 + a'_1 L^{k'} \u_{i'} + a'_2 \s_{j'}]$
share a conjugate pair, then $a_0+a'_0=1$ and, for $i \in \{1,2\}$, $a_i$ and $a'_i$
must never be simultaneously $0$.
\end{proposition}

Combining the main results in Section~\ref{sec:main} with Propositions~\ref{prop:nopair}
and~\ref{prop:pair}, the adjacency graph of $\Omega(h(x))$ can be constructed.
\begin{proposition}\label{prop:adjg}
$\Omega(h(x))$ has the following properties.
\begin{enumerate}
\item The adjacency graph of $\Omega(h(x))$ contains no loops.
\item The number of conjugate pairs between $[a_1 L^k\u_i+a_2\s_j]$ and $[\1+a'_1 L^{k'}\u_{i'}+a'_2 \s_{j'}]$
is equal to the number of conjugate pairs between $[\1+a_1 L^k\u_i+a_2 \s_j]$ and $[a'_1 L^{k'}\u_{i'}+ a'_2 \s_{j'}]$.
\item Let ${\S} \in [\1+ L^c\u_a + \s_b]$ for some $a,b,c \in \Z$. Then $[\1 + a_1 L^k \u_i + a_2 \s_j]$ and
$[a'_1 L^{k'}\u_{i'} + a'_2 \s_{j'}]$ share a conjugate pair if and only if
$
\1+a_1L^k\u_i+a_2\s_j+L^{\ell}(a'_1L^{k'}\u_{i'}+a'_2\s_{j'})$
is a shift of $\1+L^c\u_a+\s_b$ for some $\ell$ or, equivalently,
$
a_1 L^k\u_i + a_2 \s_j + L^{\ell} (a'_1 L^{k'} \u_{i'} + a'_2 \s_{j'})$
is a shift of $(0,\1) \in L^c \u_a + \s_b$ with $\1$ of length $m+n$.
\item The number of conjugate pairs between any two cycles in $\Omega(h(x))$ can be
determined completely based on Propositions \ref{prop1} to~\ref{prop5} and Algorithm~\ref{algo:b}
after small modifications.
\end{enumerate}
\end{proposition}

Propositions \ref{prop1} to~\ref{prop5} form a good foundation to study and derive the adjacency graph for $\Omega(h(x))$. Since the required modifications mentioned in the last item of Proposition~\ref{prop:adjg} are straightforward, the details are omitted here.

Once we have determined how the conjugate pairs are shared, we can perform steps analogous to those detailed in Section~\ref{sec:main} to determine the states
in a given cycle, to find conjugate pairs between any two cycles, and to estimate
the number of de Bruijn sequences constructed. The non existence of loops in the adjacency graph is an advantage.

In particular, we can construct a $(1+m+n)\times (1+m+n)$ matrix $P'$. Any state belonging to
the cycles in $\Omega(h(x))$ can then be described as $(\v_1,\v_2,\v_3)P'$ with
$\v_1 \in\F_2^1$, $\v_2 \in\F_2^m$, and $\v_3 \in\F_2^n$.
Let $\S=(1,\a_3,\b_3)P'$ and let $(1,\a_1,\b_1)P'$ and $(0,\a_2,\b_2)P'$ be the respective states
of cycles $C_1$ and $C_2$ in $\Omega(h(x))$. Run Algorithm~\ref{algo:b} on $(\a_1,\b_1)$, $(\a_2,\b_2)$, and $(\a_3,\b_3)$. If Algorithm 1 yields a conjugate pair $(\v,\hat{\v})$ with $\v =(\a',\b')P \in\F_2^{m+n}$, then $(\v'=(1,\a',\b')P',\hat{\v'})$ is the conjugate pair between $C_1$ and $C_2$.

To conclude this section, we show some advantages of our more general approach of using product of
irreducible polynomials over that of~\cite{Li16} which is limited to using primitive polynomials.

Let $\I_{2}(n)$ and $\P_{2}(n)$ denote, respectively, the number of irreducible and primitive polynomials of degree $n$ in $\F_2[x]$. We know that $\P_{2}(n)=|\A_n|=\sfrac{\phi(2^n-1)}{n}$. Let $\mu(n)$ be the M{\"o}bius function. From Gauss' general formula~\cite[Theorem 3.25]{LN97}
\[
\I_{2}(n)=\frac{1}{n} \sum_{d \mid n} \mu(d) 2^{\frac{n}{d}}.
\]
Let $N_n=\I_{2}(n)- \P_{2}(n)$. Consulting Sequences A001037 and A011260 in~\cite{OEIS} that list down
$\I_{2}(n)$ and $\P_{2}(n)$, respectively, one gets
\begin{equation*}
\setlength{\arraycolsep}{3pt}
\begin{array}{c| ccccc ccccc ccccc cccc}
n   & ~4 & 5  & 6  & 7 & 8  & 9 & 10 & 11 & 12 & 13  & 14 & 15 & 16 & 17 & 18 & 19 & 20 & 21 \\
N_n & ~1 & 0  & 3  & 0 & 14 & 8 & 39 & 10 & 191 & 0 & 405 & 382 & 2032 & 0 & 6756 & 0 & 28377 & 15186
\end{array}.
\end{equation*}
$N_n=0$ if and only if $2^n-1$ is a (Mersenne) prime. As $n$ grows larger, primes of the form $2^n-1$ appear to grow increasingly sparse. Hence, for most $n$, our method draws polynomials from a larger pool of choices than the one used in~\cite{Li16}.

Another advantage is that we get more de Bruijn sequences by using irreducible but non-primitive polynomial. The three irreducible polynomials of degree $4$ in $\F_2[x]$ are $f_1(x)=x^4+x+1$, $f_2(x)=x^4+x^3+1$, and $f_3(x)=x^4+x^3+x^2+x+1$. The first two are primitive while $f_3(x)$ is not.

Let $\Omega((1+x) f_1(x))=\{[\0],[\1],[\s],[\s+\1]\}$. Then $[\0]$ and $[\s+\1]$ as well as $[\1]$ and
$[\s]$ each share a unique conjugate pair. There are $14$ conjugate pairs shared by $[\s]$ and $[\s+\1]$. The number of sequences constructed is only $14$.

Consider $\Omega((1+x) f_3(x))=\{[\0],[\1],[\s_0],[\s_1],[\s_2],[\s_0+\1],[\s_1+\1],[\s_2+\1]\}$.
Lemma~\ref{lemma:cn} helps us compute the number of conjugate pairs shared by any two cycles. The adjacency graph is shown in Figure~\ref{figu}. By Theorem~\ref{BEST}, the number of de Bruijn sequences constructed is $576$.
\begin{figure}[h]
\vspace{-0.6cm}
\begin{tikzpicture}[auto, node distance=2.3cm, every loop/.style={},
                    thick,main node/.style={draw,font=\sffamily\bfseries}]
  \node[main node] (1) {$[\0]$};
  \node[main node] (2) [below of=1] {$[\1]$};
  \node[main node] (3) [right of=1] {$[\s_1]$};
  \node[main node] (4) [below of=3] {$[\s_1+\1]$};
  \node[main node] (5) [right of=3] {$[\s_0]$};
  \node[main node] (6) [below of=5] {$[\s_0+\1]$};
  \node[main node] (7) [right of=5] {$[\s_2]$};
  \node[main node] (8) [below of=7] {$[\s_2+\1]$};

  \path[every node/.style={font=\sffamily}]
    (1) edge [bend right] node [above] {$1$} (4)
    (2) edge node [right] {$1$} (3)
    (3) edge [bend right] node [right] {$2$} (6)
        edge node [above] {$2$} (8)
    (4) edge [bend left] node [right] {$2$} (5)
        edge node [below] {$2~~$} (7)
    (5) edge [bend right] node [left] {$2$} (6)
        edge [bend left] node [right] {$1$} (8)
    (6) edge node [right] {$1$} (7)
    (7) edge [bend left] node [right] {$2$} (8);
\end{tikzpicture}
\centering
\caption{The adjacency graph of $\Omega((1+x)(x^4+x^3+x^2+x+1))$.}
\label{figu}
\end{figure}
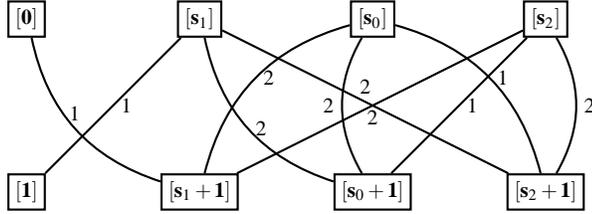
\section{Conclusion and Future Directions}\label{sec:conclu}
This paper constructs new de Bruijn sequences by the cycle joining method using products of two distinct
irreducible polynomials as characteristic polynomials. We present results on the cycle structure, provide the corresponding adjacency graph, and exhibit a connection between relevant cyclotomic numbers and the new de Bruijn sequences. Many of the results naturally extend to the case where  $f(x)= p_1(x)\cdots p_s (x)$, where $p_i(x) \in \F_2[x]$ are pairwise distinct irreducible polynomials for $1 \leq i \leq s$.

Possible applications of de Bruijn sequences merit deeper investigation. The large number of de Bruijn sequences that can be efficiently constructed based on specific choices of polynomials may be beneficial for implementations in spread spectrum, more specifically in the design of control systems for autonomous vehicles. Crucial aspects to look at in this direction are the auto and cross correlation properties of the sequences as discussed in~\cite{SG13}. Various modifications of de Bruijn sequences have been known to result in powerful tools in DNA analysis~\cite{PTW01} and DNA-based data storage systems~\cite{CCEK16}.

\begin{acknowledgements}
The work of Z.~Chang is supported by the Joint Fund of the National Natural Science Foundation of China under Grant U1304604. Research Grants TL-9014101684-01 and MOE2013-T2-1-041 support the research carried out by M.~F.~Ezerman, S.~Ling, and H.~Wang. The collaboration leading to this paper was performed while Z.~Chang was a visiting scholar at the Division of Mathematical Sciences, School of Physical and Mathematical Sciences, Nanyang Technological University.
\end{acknowledgements}

\end{document}